\documentclass[twocolumn]{autart}    

\usepackage{graphicx}          

\usepackage{cite}
\usepackage{amsmath}
\usepackage{amsfonts}
\usepackage{amssymb}

\usepackage{algorithmic}
\usepackage{graphicx}
\usepackage{textcomp}
\usepackage{booktabs}
\usepackage{xcolor}

\newtheorem{theorem}{Theorem}
\newtheorem{lemma}{Lemma}

\newenvironment{proof}{{\itshape Proof.}}{\hspace*{0pt}\hfill$\qed$}
\newcounter{remarkcounter}
\newenvironment{remark}{\refstepcounter{remarkcounter}{\bfseries Remark~\arabic{remarkcounter}.}}{\hspace*{0pt}\hfill$\blacksquare$}

\allowdisplaybreaks
\begin{document}

\begin{frontmatter}

\title{Distributed Observers with Dynamic Event-Triggered Communication} 

\thanks[footnoteinfo]{This paper was not presented at any IFAC meeting. This work was supported in part by the Natural Science Foundation of China under Grants 62273227 and 92367203. Corresponding author Xianwei Li.}

\author[a]{Yiyang Liu}\ead{lyy237000@sjtu.edu.cn},    
\author[a,b,c]{Xianwei Li}\ead{xianwei.li@sjtu.edu.cn},               
\author[a,b]{Shaoyuan Li}\ead{syli@sjtu.edu.cn}  

\address[a]{School of Automation and Intelligent Sensing, Shanghai Jiao Tong University, Shanghai 200240, China}  
\address[b]{Key Laboratory of System Control and Information Processing, Ministry of Education of China, Shanghai 200240,China}
\address[c]{State Key Laboratory of Submarine Geoscience, Shanghai Jiao Tong University, Shanghai 200240, China}             

\begin{keyword}                           
Distributed observers, dynamic event-triggering mechanisms, linear time-invariant (LTI) systems, minimum inter-event times (MIETs).               
\end{keyword}                             

\begin{abstract}                          
This paper studies the problem of distributed state estimation of linear time-invariant (LTI) systems under event-triggered communication. For event-triggering mechanisms, the existence of positive minimum inter-event times (MIETs) is an essential property for ensuring practicality. It is widely recognized that dynamic event-triggering mechanisms can effectively reduce redundant communication. However, for distributed observers, it remains unclear whether dynamic event-triggering mechanisms can ensure positive MIETs. This paper proposes a dynamic event-triggered distributed observer. By introducing new comparison functions, it is proven that the dynamic event-triggered distributed observer can guarantee strictly positive MIETs and ensure the exponential convergence of the estimation error. Moreover, most existing works on event-triggered distributed observers only consider node-based event-triggering mechanisms, while both node-based and edge-based dynamic event-triggering mechanisms are constructed in this paper. Numerical examples are provided to illustrate the effectiveness of the proposed results.
\end{abstract}

\end{frontmatter}

\section{Introduction}
State estimation for dynamic systems has long been a fundamental problem in modern control theory. For large-scale dynamic systems, due to the limited observation capabilities of a single agent, it is common to deploy multiple agents to collectively estimate the system state, i.e., distributed state estimation. Typical applications of distributed state estimation include target tracking in sensor networks \cite{petitti2011consensus}, environmental monitoring \cite{park2016optimal}, power network monitoring \cite{pasqualetti2012distributed}, multiagent consensus tracking \cite{lv2019adaptive,lv2020adaptive} etc.

There are two primary approaches to distributed state estimation: distributed Kalman filter and distributed observer. The distributed Kalman filter exchanges information with neighboring nodes and iteratively updates both state estimates and covariance matrices (see \cite{olfati2005distributed,olfati2007distributed,olfati2009kalman,kamgarpour2008convergence,matei2012consensus,kim2016distributed} for details). Although the distributed Kalman filter demonstrates superior performance, it suffers from high computational complexity.

In contrast to distributed Kalman filter, a distributed observer stabilizes estimation errors using fewer computational resources. The main challenge in the distributed observer-based approach is that each local observer can only access partial system information, making it unable to recover the full state on its own. To tackle this challenge, a state-augmented distributed observer is proposed in \cite{park2016design}, where the distributed estimation problem is transformed into a decentralized stabilization problem. The authors in \cite{wang2017distributed} propose a simpler state-augmented observer than the one in \cite{park2016design}, featuring a tunable error convergence rate. Unlike state-augmented observers, the methods proposed in \cite{zhu2014cooperative,mitra2018distributed} avoid state augmentation. These observers consist of a Luenberger-type local observer part and a  consensus part. In \cite{zhu2014cooperative}, the output injection gain and coupling gain are co-designed to achieve state omniscience. In \cite{mitra2018distributed}, the multi-sensor observable canonical decomposition is proposed, by which each agent is associated with an observable substate and an unobservable substate. The observable substate is estimated by a Luenberger-type observer, while the unobservable substate is corrected by a consensus-based update rule. Notably, although the aforementioned observers operate in a distributed fashion, their design is centralized. To address this limitation, the work \cite{kim2019completely} proposes a fully distributed observer whose design requires no global information. Building upon the foundation laid in \cite{kim2019completely}, several extensions are developed in \cite{wang2020consensus,zhang2023distributed,zhang2025distributed}. In \cite{wang2020consensus}, a distributed reduce-order observer is proposed. The distributed state estimation problem under switching communication networks was investigated in \cite{zhang2023distributed}. Subsequently, this work was extended to additionally consider modeling uncertainty in \cite{zhang2025distributed}.

Note that all the aforementioned approaches require persistent communication between agents, incurring great communication burden. To reduce communication burden, event-triggering mechanisms offer an efficient solution \cite{heemels2012introduction,dimarogonas2011distributed,zhan2025dynamic,girard2014dynamic,qian2018output,qian2021design,meng2013event,meng2015periodic}. In particular, agents sample and transmit their states to neighbors only when certain predefined condition is satisfied. Event-triggering mechanisms can be broadly classified into two main categories: static event-triggering mechanisms and dynamic event-triggering mechanisms. It has been indicated in \cite{girard2014dynamic,dolk2017output,ding2018overview} that the dynamic event-triggering mechanisms can reduce more redundant information transmissions without compromising performance, when compared with static ones. Over the past few years, much interest has been drawn by dynamic event-triggered distributed observers \cite{liu2021resilient,su2024adaptive,long2025adaptive,zhu2025hybrid}.

A key consideration for event-triggering mechanisms is Zeno behavior, which refers to the the phenomenon of infinite triggers occurring in finite time and thus must be avoided. In \cite{liu2021resilient,su2024adaptive,long2025adaptive,sun2020event}, the event-triggered distributed observers are proven to be Zeno-free. However, the standard notion of Zeno-freeness is not concerned with the limiting behavior of inter-event times as time goes to infinity. In other words, even if an event-triggering mechanism is proven to be Zeno-free, its inter-event intervals may still decrease over time and even converge to zero in the limit \cite{nowzari2019event}. Thus, merely ensuring Zeno-freeness is insufficient for practicality of event-triggered sampling algorithms. A stronger and more practical property than Zeno-freeness is the existence of a positive lower bound of minimum inter-event times (MIETs). To the best of our knowledge, aside from the works \cite{zhu2025hybrid,li2025distributed}, there are quite few works on event-triggered distributed observers which can ensure positive MIETs. In \cite{zhu2025hybrid,li2025distributed}, to achieve a guarantee of positive MIETs, either time-regularization or periodic checking technique is employed therein. However, the previous works without or with positive MIETs guarantees \cite{liu2021resilient,su2024adaptive,long2025adaptive,zhu2025hybrid,sun2020event,li2025distributed} have overlooked whether the original dynamic event-triggering mechanism can ensure positive MIETs. \textit{Although it is well-known that centralized dynamic event-triggered stabilization can ensure positive MIETs \cite{girard2014dynamic}, it is unclear whether this property can be achieved for distributed dynamic event-triggered state estimation. Addressing this gap is one of the motivations for this work.}  

Notably, the event-triggering mechanisms presented in the aforementioned works \cite{sun2020event,liu2021resilient,su2024adaptive,long2025adaptive,zhu2025hybrid} are node-based, which assumes that once an event is triggered for an agent, all its neighbors can receive the sampled information simultaneously. Node-based event-triggering mechanisms are suitable for communication protocols through broadcasting, but are not compatible with communication protocols through one-to-one connection. For the latter case, edge-based event-triggering mechanisms are a feasible solution. An edge-based event-triggered distributed observer is proposed in \cite{li2025distributed}. By setting triggering conditions for each communication link between agents, each agent can communicate asynchronously with its neighbors. However, in \cite{li2025distributed}, the observer design requires solving a linear matrix inequality (LMI), whose feasibility is not guaranteed. Specifically, a necessary condition for this inequality to be feasible is that the undetectable subsystem is stable. Besides, in \cite{li2025distributed}, the edges $(i,j)$ and $(j,i)$ are required to be triggered simultaneously. \textit{To the best of our knowledge, aside from \cite{li2025distributed}, research on the edge-based event-triggered distributed observer remains limited, which is another motivation of this work.}

%

In this work, dynamic event-triggered distributed observers are designed, which includes node-based and edge-based type. Moreover, we prove that the proposed observers can ensure exponential convergence of the estimation error and strictly positive MIETs. Compared with related works, the main contributions of this paper are summarized as follows:
\begin{enumerate}
	\item[1)] Dynamic event-triggered distributed observers with positive MIET guarantees are proposed. Compared with related works \cite{sun2020event,liu2021resilient,su2024adaptive,long2025adaptive} which merely exclude Zeno behavior, the existence of a positive lower bound of MIETs provides more guarantees than Zeno-freeness for reliably implementing an event-triggering mechanism. Notably, in \cite{zhu2025hybrid,li2025distributed}, strictly positive MIETs are ensured either by time-regularization or by adopting a periodic event-triggering mechanism, whereas we prove that the existence of the positive MIETs is an intrinsic property of the designed dynamic event-triggering mechanisms. 
	\item[2)] Both node-based and edge-based event-triggering mechanisms are designed. Apart from \cite{li2025distributed}, there is limited work on edge-based event-triggered distributed observers. Notably, the observer design in \cite{li2025distributed} entails solving a LMI, the feasibility of which necessarily requires the stability of the undetectable subsystem. Our method, by contrast, does not suffer from this limitation. Additionally, different from \cite{li2025distributed} where edges $(i,j)$ and $(j,i)$ are triggered synchronously, in our method, all edges are triggered asynchronously, which is more practical.
\end{enumerate}

This paper is organized as follows: Section \ref{Preliminaries} states preliminaries and problem formulation. Section \ref{Node-based Event-Triggered Distributed Observer} and Section \ref{Edge-based Event-Triggered Distributed Observer} present node-based and edge-based event-triggered distributed observers, respectively. Section \ref{Numerical Examples} presents the simulation results. Section \ref{Conclusion} concludes this paper and discusses future research directions.

\section{Preliminaries and Problem Formulation} \label{Preliminaries}
\subsection{Notation}
In this paper, denote $I_n$ as the identity matrix of dimension $n\times n$. Denote $\mathbb{R}^{m\times n}$ as the set of $m\times n$ dimensional real matrices. For a  matrix $A\in \mathbb{R}^{m\times n}$, the kernel of $A$ is denoted as $\mathrm{ker}(A)\triangleq \{x\in\mathbb{R}^n | Ax=0 \}$, the image of $A$ is denoted as $\mathrm{im}(A)\triangleq \{y\in\mathbb{R}^m|y=Ax, \exists x\in\mathbb{R}^n \}$, and the induced matrix 2-norm of $A$ is denoted as $||A||$. For a symmetric matrix $P$, $\lambda_{\mathrm{min}}$ and $\lambda_{\mathrm{max}}$ denote its minimum and maximum eigenvalue, respectively. For a vector $x$, $||x||$ denotes its 2-norm. For a subspace $\mathcal{Z} \subset \mathbb{R}^n$, the orthogonal complement of $\mathcal{Z}$ is denoted as $\mathcal{Z}^\bot \triangleq \{x\in\mathbb{R}^n | x^\mathrm{T} z = 0, \forall z\in\mathcal{Z}\}$. The symbol $\lor$ denotes the logical OR operation. Assume that a set of matrices $\{R_i | i = 1,2,...,m \}$ and an index set $\mathcal{M}=\{1,2,...,m\}$ are given. We define a concatenated matrix $\mathrm{col}\{R_i\}_{i\in\mathcal{M}}\triangleq \begin{bmatrix}
	R_1^\mathrm{T} & R_2^\mathrm{T} & \cdots & R_m^\mathrm{T} 
\end{bmatrix}^\mathrm{T}$ and let $\mathrm{diag}\{R_i\}_{i\in \mathcal{M}}$ denote the block-diagonal matrix with $R_i$, $i\in \mathcal{M}$ as its diagonal blocks.

\subsection{Graph Theory}
To describe the communication topology among the agents, an undirected graph $\mathcal{G}(\mathcal{V},\mathcal{E},\mathcal{A})$ is employed in this paper, where $\mathcal{V}\triangleq \{1,2,...,N\}$ denotes the node set, $\mathcal{E} \subseteq \mathcal{V}\times\mathcal{V}$ denotes the edge set, and $\mathcal{A}\triangleq[a_{ij}]_{N\times N}$ denotes the adjacency matrix. If $(j,i)\in\mathcal{E}$, it means that agent $i$ can communicate with agent $j$ and $(j,i)\notin\mathcal{E}$ otherwise. Besides, self-connection is excluded, i.e., $(i,i)\notin \mathcal{E}$. Define the neighbor set of agent $i$ as $\mathcal{N}_i\triangleq \{j\in\mathcal{V} | (j,i)\in\mathcal{E}\}$. Let $N_i$ be the number of agents in $\mathcal{N}_i$. For adjacency matrix $\mathcal{A}$, $a_{ij}>0$ if $j\in\mathcal{N}_i$ and $a_{ij}=0$ otherwise. Since the graph is undirected, we have $a_{ij}=a_{ji}$. Denote $\mathcal{L}\triangleq[l_{ij}]_{N\times N}$ as the Laplacian matrix, where $l_{ii}=\sum_{k\in\mathcal{N}_i} a_{ik}$ for $i\in\mathcal{V}$ and $l_{ij}=-a_{ij}$ for $i,j\in\mathcal{V}$ and $i\neq j$.

\subsection{Problem Formulation}
In this paper, we consider a linear time-invariant system 
\begin{align}
	\dot{x} = Ax, \label{LTI system}
\end{align}
where $A\in \mathbb{R}^{n\times n}$ is the system matrix, and $x\in \mathbb{R}^n$ is the system state. The system state is monitored by $N$ agents.  
Each agent can only obtain partial information about the system state due to the limited sensing capabilities. Specifically, agent $i$ receives a partial measurement of the state in form of
\begin{align}
	y_i = H_i x, \label{y_i}
\end{align}
where $H_i \in \mathbb{R}^{m_i \times n}$ is the measurement output matrix.
By stacking each local measurement $y_i$, the entire measurement model is formulated as
\begin{align}
	\begin{bmatrix}
		y_1 \\
		y_2 \\
		\vdots \\
		y_N
	\end{bmatrix} = \begin{bmatrix}
		H_1 \\
		H_2 \\
		\vdots \\
		H_N
	\end{bmatrix}x \triangleq Hx. \label{y}
\end{align}
In this paper, the pair $(A,H)$ is assumed to be detectable while the pair $(A,H_i)$ is not required to be detectable.


Each agent maintains an estimate $\hat{x}_i(t)$ of the system state. It is assumed that each agent  communicates with its neighbors according to an undirected communication topology described by $\mathcal{G}$. The aim of this paper is to design a distributed observer that enables each agent to asymptotically reconstruct the system state, i.e.,
$	\lim_{t\rightarrow \infty} (\hat{x}_i(t) - x(t)) = 0$ for $i=1,2,...,N$. Most existing works are based on continuous information interaction between agents (e.g., \cite{kim2019completely,wang2020consensus,zhang2023distributed,zhang2025distributed}), which requires persistent inter-agent communication and thus brings about great communication burdens. In this paper, to reduce the communication burden, we will design dynamic event-triggering mechanisms to govern the information transmission of each agent. In particular, it is expected that the designed dynamic event-triggering mechanism can guarantee strictly positive MIETs in order to keep its practicality.

\subsection{Detectability Decomposition} \label{detectablity decomposition section}
We first recall the definition of undetectable subspace of the pair $(A,H_i)$ and introduce some notation for later use. Denote the characteristic polynomial of $A$ as $\mathcal{F}_A(s)$, factored as $\mathcal{F}_A(s)= \mathcal{F}_A^-(s)\mathcal{F}_A^+(s)$, where the roots of $\mathcal{F}_A^-(s)$ and $\mathcal{F}_A^+(s)$ are in the open left and the closed right half-planes of the complex plane, respectively. Then the undetectable subspace of the pair $(A,H_i)$ is defined as
\begin{align}
	\mathcal{U}_i \triangleq \cap_{k=1}^n \mathrm{ker}(H_i A^{k-1})\cap\mathrm{ker}(\mathcal{F}_A^+(A)), \nonumber
\end{align}
whose dimension is denoted as $p_i$.

The pair $(A,H_i)$ is said to be detectable if $\mathcal{U}_i=\{0\}$.
Define $D_i \in \mathbb{R}^{n\times(n-p_i)}$ and $U_i \in \mathbb{R}^{n\times p_i}$ as the matrices whose columns are the orthonormal basis of $\mathcal{U}_i^\bot$ and $\mathcal{U}_i$,  respectively. Then it follows that 
\begin{align}
	\mathcal{U}_i = \mathrm{im}(U_i) = \mathrm{ker}(D_i^\mathrm{T}). \nonumber
\end{align}

Here, we recall the detectability decomposition for the pair $(A,H_i)$. Define an orthonormal matrix $T_i$ as
\begin{align}
	T_i \triangleq \begin{bmatrix}
		D_i & U_i
	\end{bmatrix}. \label{transformation matrix}
\end{align}
The pair $(A, H_i)$ can be transformed by the change of basis matrix $T_i$ into the form of
\begin{align}
	T_i^\mathrm{T} A T_i = \begin{bmatrix}
		A_{i\mathrm{d}} & 0 \\
		A_{i\mathrm{r}} & A_{i\mathrm{u}}
	\end{bmatrix}, \,\, H_i T_i = \begin{bmatrix}
		H_{i\mathrm{d}} & 0
	\end{bmatrix}, \nonumber
\end{align}
where $A_{i\mathrm{d}}\in \mathbb{R}^{(n-p_i)\times(n-p_i)}$, $A_{i\mathrm{r}} \in\mathbb{R}^{p_i \times (n-p_i)}$, $A_{i\mathrm{u}}\in \mathbb{R}^{p_i\times p_i}$ and $H_{i\mathrm{d}} \in \mathbb{R}^{m_i \times (n-p_i)}$. Moreover, the pair $(A_{i\mathrm{d}},H_{i\mathrm{d}})$ is detectable.

\begin{lemma} [\!\!\cite{kim2019completely}] \label{lemma 2.1}
	If the communication graph $\mathcal{G}$ is connected and the pair $(A, H)$ is detectable, the matrix $U^{\mathrm{T}} (\mathcal{L} \otimes I_n)U$ is positive definite, where $U = \mathrm{diag}\{U_i\}_{i\in\mathcal{V}}$.
\end{lemma}

\section{Node-based Event-Triggered Distributed Observer} \label{Node-based Event-Triggered Distributed Observer}
\subsection{Observer Design}

We first design a dynamic event-triggered distributed observer for the system (\ref{LTI system}) with measurement (\ref{y}). It consists of $N$ local observers, and the local observer of agent $i$ has the dynamics
\begin{align}
	\dot{\hat{x}}_i (t) &= A \hat{x}_i(t) + L_i (y_i(t) - H_i \hat{x}_i(t)) \nonumber \\
	&\quad - cM_i \sum_{j\in \mathcal{N}_i} a_{ij} (\bar{x}_i(t)-\bar{x}_j(t)), \label{observer n}
\end{align}
where $\bar{x}_i(t) = \mathrm{e}^{A(t-t_k^i)} \hat{x}_i(t_k^i)$, $t \in \left[t_k^i,t_{k+1}^i\right)$, with $\{t_k^i\}_{k\in\mathbb{N}}$ being the sampling instant sequence for agent $i$. Moreover, the gain $c$ is chosen as any constant satisfying
\begin{align}
	c \geq \frac{2(2 + \|A_\mathrm{u}+A_\mathrm{u}^\mathrm{T}\|)}{\lambda_{\mathrm{min}}(U^\mathrm{T} (\mathcal{L}\otimes I_n) U)}, \label{choice of c n}
\end{align} 
where $A_\mathrm{u} = \mathrm{diag}\{A_{i\mathrm{u}}\}_{i\in\mathcal{V}}$ and the graph $\mathcal{G}$ is assumed to be connected so that the right hand side of (\ref{choice of c n}) is a finite constant by Lemma \ref{lemma 2.1}.
The gain matrices $L_i$ and $M_i$ are designed as
\begin{align}
	L_i &= T_i \begin{bmatrix}
		L_{i\mathrm{d}} \\
		0
	\end{bmatrix} , \,\,
	M_i = T_i \begin{bmatrix}
		0 & 0\\
		0 & I_{p_i}
	\end{bmatrix} T_i^\mathrm{T} , \label{gain matrix}
\end{align}
where $T_i$ is the transformation matrix given in (\ref{transformation matrix}), $p_i$ is the dimension of the undetectable subspace of $(A,H_i)$, and $L_{i\mathrm{d}}$ is chosen such that $A_{i\mathrm{d}}-L_{i\mathrm{d}}H_{i\mathrm{d}}$ is Hurwitz. Therefore, there always exists a positive definite matrix $P_{i\mathrm{d}}$ such that 
\begin{align}
	(A_{i\mathrm{d}}-L_{i\mathrm{d}}H_{i\mathrm{d}})^\mathrm{T} P_{i\mathrm{d}} + P_{i\mathrm{d}} (A_{i\mathrm{d}}-L_{i\mathrm{d}}H_{i\mathrm{d}}) = -I_{n-p_i}. \label{Riccati equation}
\end{align}

We now present the node-based event-triggering mechanism. Define $\tilde{x}_i = \bar{x}_i - \hat{x}_i$ and perform the detectability decomposition on $\tilde{x}_i$ using $T_i$ as $\tilde{x}_i = T_i \begin{bmatrix}
	\tilde{x}_{i\mathrm{d}} \\
	\tilde{x}_{i\mathrm{u}}
\end{bmatrix}$. The sampling instant sequence $\{t_k^i\}_{k\in\mathbb{N}}$ is generated by the following dynamic event-triggering mechanism:
\begin{align}
	t_{k+1}^i = \inf \{t>t_k^i|f_i(t) \geq 0 \lor t-t_k^i \geq \bar{\tau}_i\},  \quad t_0^i = 0  , \label{ET condition n}
\end{align}
where
\begin{align}
	f_i(t) &= 4l_{ii} \| \tilde{x}_{i\mathrm{u}} \|^2  - \beta \sum_{j \in \mathcal{N}_i} a_{ij} \| \bar{x}_i -  \bar{x}_j \|^2 - \kappa_i \rho_i(t), \label{ET function n}
\end{align}
and $\rho_i(\cdot) : \left[0,\infty\right) \rightarrow \left(0,\infty\right)$ evolves according to
\begin{align}
	\dot{\rho}_i &= -\delta_i \rho_i - 4l_{ii} \| \tilde{x}_{i\mathrm{u}} \|^2   \nonumber \\
	&\quad     +  \beta \sum_{j \in \mathcal{N}_i} a_{ij} \| \bar{x}_i -  \bar{x}_j \|^2 + \gamma_i \|H_i \hat{x}_i-y_i\|^2. \label{dynamic variable n} 
\end{align}
In (\ref{ET condition n})--(\ref{dynamic variable n}), $\bar{\tau}_i$, $\delta_i$, $\kappa_i$, $\gamma_i$ and $\rho_i(0)$ are any positive constants, and $0<\beta<1$. Note that the triggering condition (\ref{ET condition n}) ensures that $f_i(t)<0$ for all $ t\geq 0$. Therefore, it is easy to show that $\rho_i(t) >0$ for all $t\geq 0$.

\begin{remark}
	The observer design in (\ref{observer n}) is inspired by \cite{kim2019completely}. Different from the work \cite{kim2019completely}, the consensus error term in (\ref{observer n}) consists of sampled state estimates. Different from other existing works \cite{sun2020event,liu2021resilient,su2024adaptive,long2025adaptive,zhu2025hybrid,li2025distributed}, the event-triggering mechanism (\ref{ET condition n}) incorporates a variable $\rho_i$ whose dynamics are related to the error information in terms of $\tilde{x}_{i\mathrm{u}}$ and $\|\bar{x}_i - \bar{x}_j\|^2$. As shown later, introducing these terms is crucial for proving the existence of strictly positive MIETs.
\end{remark}

\subsection{Convergence Analysis}

In this section, we prove that the observer (\ref{observer n}) together with the dynamic event-triggering mechanism (\ref{ET condition n}) guarantees the exponential convergence of the estimation error $\eta_i \triangleq \hat{x}_i-x$.

\begin{theorem} \label{theorem 3.1}
	Consider the system (\ref{LTI system}) and (\ref{y}) over the communication network $\mathcal{G}$. Suppose that $\mathcal{G}$ is undirected and connected, and the pair $(A,H)$ is detectable. Then the observer (\ref{observer n})--(\ref{gain matrix}) with the node-based dynamic event-triggering mechanism (\ref{ET condition n}) guarantees that the estimation error $\eta_i$ converges to zero exponentially.
\end{theorem}

\begin{proof}
	We first formulate the dynamics for the detectable and undetectable part of the estimation error. The estimation error $\eta_i$ satisfies
	\begin{align}
		\dot{\eta}_i &= (A-L_i H_i)\eta_i-cM_i \sum_{j\in \mathcal{N}_i} a_{ij} (\bar{x}_i-\bar{x}_j) \nonumber \\
		&= (A-L_i H_i)\eta_i-cM_i \sum_{j\in \mathcal{N}_i} a_{ij} (\bar{\eta}_i-\bar{\eta}_j),
		\label{dot eta i n}
	\end{align}
	where $\bar{\eta}_i = \bar{x}_i - x$. Performing the detectability decomposition on $\eta_i$ using $T_i$ yields $\eta_i = T_i \begin{bmatrix}
		\eta_{i\mathrm{d}} \\
		\eta_{i\mathrm{u}}
	\end{bmatrix}$. By (\ref{dot eta i n}), the dynamics of $\eta_{i\mathrm{d}}$ and $\eta_{i\mathrm{u}}$ can be obtained:
	\begin{align}
		\dot{\eta}_{i\mathrm{d}} &= (A_{i\mathrm{d}} - L_{i\mathrm{d}} H_{i\mathrm{d}} ) \eta_{i\mathrm{d}}, \label{dot eta id n} \\
		\dot{\eta}_{i\mathrm{u}} &= A_{i\mathrm{r}} \eta_{i\mathrm{d}} + A_{i\mathrm{u}} \eta_{i\mathrm{u}} - cU_i^\mathrm{T} \sum_{j \in \mathcal{N}_i} a_{ij}(\bar{\eta}_i - \bar{\eta}_j) \nonumber \\
		&= A_{i\mathrm{r}} \eta_{i\mathrm{d}} + A_{i\mathrm{u}} \eta_{i\mathrm{u}} - cU_i^\mathrm{T} \sum_{j =1}^{N} l_{ij}\left(D_j \bar{\eta}_{j\mathrm{d}} + U_j \bar{\eta}_{j\mathrm{u}} \right) , \label{dot eta iu n}
	\end{align}
	where 
	\begin{align}
		\bar{\eta}_i =  T_i\begin{bmatrix}
			\bar{\eta}_{i\mathrm{d}} \\
			\bar{\eta}_{i\mathrm{u}}
		\end{bmatrix}. \label{bar eta i n}
	\end{align} 
	
	Let $\tilde{\eta}_i = \bar{\eta}_i - \eta_i = T_i \begin{bmatrix}
		\tilde{\eta}_{i\mathrm{d}} \\
		\tilde{\eta}_{i\mathrm{u}}
	\end{bmatrix}$. It is easy to obtain that 
	\begin{equation}
		\tilde{\eta}_i=\tilde{x}_i,\,\,\tilde{\eta}_{i\mathrm{d}}=\tilde{x}_{i\mathrm{d}},\,\,\tilde{\eta}_{i\mathrm{u}}=\tilde{x}_{i\mathrm{u}}. \nonumber
	\end{equation}
	Define the concatenated variables 
	\begin{equation}
		\begin{aligned}
			&\eta_\mathrm{d} = \mathrm{col}\{\eta_{i\mathrm{d}}\}_{i \in \mathcal{V}}, \quad \eta_\mathrm{u} = \mathrm{col}\{\eta_{i\mathrm{u}}\}_{i \in \mathcal{V}}, \\
			&\tilde{\eta}_\mathrm{d} = \mathrm{col}\{\tilde{\eta}_{i\mathrm{d}}\}_{i \in \mathcal{V}}, \quad \tilde{\eta}_\mathrm{u} = \mathrm{col}\{\tilde{\eta}_{i\mathrm{u}}\}_{i \in \mathcal{V}},
		\end{aligned} \label{concatenated variables}
	\end{equation}
	and the concatenated matrices $A_\mathrm{d} = \mathrm{diag}\{A_{i\mathrm{d}}\}_{i \in \mathcal{V}}$, $A_\mathrm{r} = \mathrm{diag}\{A_{i\mathrm{r}}\}_{i \in \mathcal{V}}$,  $L_\mathrm{d} = \mathrm{diag}\{ L_{i\mathrm{d}} \}_{i \in \mathcal{V}}$, $H_\mathrm{d} = \mathrm{diag}\{ H_{i\mathrm{d}} \}_{i \in \mathcal{V}}$, $D = \mathrm{diag}\{ D_i \}_{i \in \mathcal{V}}$, $P_\mathrm{d} = \mathrm{diag}\{P_{i\mathrm{d}}\}_{i\in\mathcal{V}}$.
	
		
		We now analyze the convergence of the error dynamics given in (\ref{dot eta id n}) and (\ref{dot eta iu n}). Choose the following candidate Lyapunov function:
		\begin{align}
			V =  \mu_\mathrm{d} \eta_{\mathrm{d}}^\mathrm{T} P_\mathrm{d} \eta_\mathrm{d} +   \eta_{\mathrm{u}}^\mathrm{T} \eta_{\mathrm{u}} + \frac{c}{2} \sum_{i\in\mathcal{V}} \rho_i, \label{Lyapunov function n}
		\end{align} 
		where $P_{\mathrm{d}}$ is a positive-definite matrix satisfying (\ref{Riccati equation}) and $\mu_\mathrm{d}$ is a positive constant to be determined later.
		
		Along the solution to (\ref{dot eta id n}) and (\ref{dot eta iu n}) and by (\ref{Riccati equation}), the time derivative of $V$ is given by
		\begin{align}
			\dot{V} &= -\mu_\mathrm{d} \eta_{\mathrm{d}}^\mathrm{T} \eta_{\mathrm{d}} + \frac{c}{2} \sum_{i\in\mathcal{V}} \dot{\rho}_i + \sum_{i\in\mathcal{V}} 2\eta_{i\mathrm{u}}^\mathrm{T} \Big[ A_{i\mathrm{r}} \eta_{i\mathrm{d}} + A_{i\mathrm{u}} \eta_{i\mathrm{u}} \nonumber \\
			&\quad - cU_i^\mathrm{T} \sum_{j =1}^{N} l_{ij} \left(D_j \bar{\eta}_{j\mathrm{d}} + U_j \bar{\eta}_{j\mathrm{u}} \right)   \Big]  \nonumber \\
			&= -\mu_\mathrm{d} \eta_{\mathrm{d}}^\mathrm{T} \eta_{\mathrm{d}} + \frac{c}{2} \sum_{i\in\mathcal{V}} \dot{\rho}_i  + 2 \eta_\mathrm{u}^\mathrm{T} A_\mathrm{r} \eta_\mathrm{d} + 2 \eta_\mathrm{u}^\mathrm{T} A_\mathrm{u} \eta_\mathrm{u} \nonumber \\
			&\quad  - 2c \eta_\mathrm{u}^\mathrm{T} U^\mathrm{T} (\mathcal{L} \otimes I_n) U \bar{\eta}_\mathrm{u} - 2c \eta_\mathrm{u}^\mathrm{T} U^\mathrm{T} (\mathcal{L} \otimes I_n) D \bar{\eta}_\mathrm{d}. \nonumber 
		\end{align}
		According to Young's inequality, it follows that 
		\begin{align}
			\dot{V}	&\leq -\mu_\mathrm{d} \eta_{\mathrm{d}}^\mathrm{T} \eta_{\mathrm{d}} + \eta_\mathrm{d}^\mathrm{T} A_\mathrm{r}^\mathrm{T} A_\mathrm{r} \eta_\mathrm{d} + \eta_\mathrm{u}^\mathrm{T} \eta_\mathrm{u}  + 2 \eta_\mathrm{u}^\mathrm{T} A_\mathrm{u} \eta_\mathrm{u}   \nonumber \\
			&\quad  - 2c \eta_\mathrm{u}^\mathrm{T} U^\mathrm{T} (\mathcal{L} \otimes I_n) U \bar{\eta}_\mathrm{u} + \frac{c}{2} \eta_\mathrm{u}^\mathrm{T} U^\mathrm{T} (\mathcal{L}\otimes I_n) U \eta_\mathrm{u}  \nonumber \\
			&\quad + 2c \bar{\eta}_\mathrm{d}^\mathrm{T} D^\mathrm{T} (\mathcal{L}\otimes I_n) D \bar{\eta}_\mathrm{d} + \frac{c}{2} \sum_{i\in\mathcal{V}} \dot{\rho}_i. \label{dot V}
		\end{align}
		
		Note that 
		\begin{align}
			&\tilde{\eta}_\mathrm{u}^\mathrm{T} U^\mathrm{T} (\mathcal{L}\otimes I_n) U \tilde{\eta}_\mathrm{u} \nonumber \\
			&= (\bar{\eta}_\mathrm{u} - \eta_\mathrm{u})^\mathrm{T} U^\mathrm{T} (\mathcal{L}\otimes I_n) U (\bar{\eta}_\mathrm{u} - \eta_\mathrm{u}) \nonumber \\
			&= \bar{\eta}_\mathrm{u}^\mathrm{T} U^\mathrm{T} (\mathcal{L}\otimes I_n) U \bar{\eta}_\mathrm{u} +  \eta_\mathrm{u}^\mathrm{T} U^\mathrm{T} (\mathcal{L}\otimes I_n) U \eta_\mathrm{u} \nonumber \\
			&\quad - 2\eta_\mathrm{u}^\mathrm{T} U^\mathrm{T} (\mathcal{L}\otimes I_n) U \bar{\eta}_\mathrm{u}.  \label{eq1 in proof 1}
		\end{align}
		
		In light of $a_{ij} = a_{ji}$ and $U_i^\mathrm{T}U_i = I_{p_i}$, we have
		\begin{align}
			&\tilde{\eta}_\mathrm{u}^\mathrm{T} U^\mathrm{T} (\mathcal{L}\otimes I_n) U \tilde{\eta}_\mathrm{u} \nonumber \\
			&= \frac{1}{2} \sum_{i\in\mathcal{V}} \sum_{j\in\mathcal{N}_i} a_{ij} \Vert U_i \tilde{\eta}_{i\mathrm{u}} -  U_j \tilde{\eta}_{j\mathrm{u}}  \Vert^2 \nonumber \\
			&\leq \sum_{i\in\mathcal{V}} \sum_{j\in\mathcal{N}_i} a_{ij} (\Vert U_i \tilde{\eta}_{i\mathrm{u}} \Vert^2 + \Vert U_j \tilde{\eta}_{j\mathrm{u}}  \Vert^2) \nonumber \\
			&= 2 \sum_{i\in\mathcal{V}} \sum_{j\in\mathcal{N}_i} a_{ij} \Vert U_i \tilde{\eta}_{i\mathrm{u}} \Vert^2 = 2 \sum_{i\in\mathcal{V}} l_{ii} \Vert \tilde{\eta}_{i\mathrm{u}} \Vert^2, \label{ineq2 in proof 1}
		\end{align}
		and 
		\begin{align}
			&\bar{\eta}_\mathrm{u}^\mathrm{T} U^\mathrm{T} (\mathcal{L}\otimes I_n) U \bar{\eta}_\mathrm{u}  \nonumber \\
			&= (\bar{\eta}-D\bar{\eta}_\mathrm{d})^\mathrm{T} (\mathcal{L} \otimes I_n)  (\bar{\eta}-D\bar{\eta}_\mathrm{d}) \nonumber \\
			&= \bar{\eta}^\mathrm{T} (\mathcal{L}\otimes I_n) \bar{\eta} + \bar{\eta}_\mathrm{d}^\mathrm{T} D^\mathrm{T} (\mathcal{L}\otimes I_n) D \bar{\eta}_\mathrm{d} \nonumber \\
			&\quad - 2\bar{\eta}^\mathrm{T} (\mathcal{L}\otimes I_n) D\bar{\eta}_\mathrm{d} \nonumber \\
			&\geq \bar{\eta}^\mathrm{T} (\mathcal{L}\otimes I_n) \bar{\eta} + \bar{\eta}_\mathrm{d}^\mathrm{T} D^\mathrm{T} (\mathcal{L}\otimes I_n) D \bar{\eta}_\mathrm{d} \nonumber \\
			&\quad - (1-\beta) \bar{\eta}^\mathrm{T} (\mathcal{L}\otimes I_n) \bar{\eta} - \frac{1}{1-\beta} \bar{\eta}_\mathrm{d}^\mathrm{T} D^\mathrm{T} (\mathcal{L}\otimes I_n) D \bar{\eta}_\mathrm{d}  \nonumber \\
			&= \beta \bar{\eta}^\mathrm{T} (\mathcal{L}\otimes I_n) \bar{\eta} - \frac{\beta}{1-\beta} \bar{\eta}_\mathrm{d}^\mathrm{T} D^\mathrm{T} (\mathcal{L}\otimes I_n) D \bar{\eta}_\mathrm{d} \nonumber \\
			&= \frac{\beta}{2} \sum_{i\in\mathcal{V}}\sum_{j\in\mathcal{N}_i} a_{ij} \| \bar{\eta}_i -  \bar{\eta}_j \|^2 - \frac{\beta}{1-\beta} \bar{\eta}_\mathrm{d}^\mathrm{T} D^\mathrm{T} (\mathcal{L}\otimes I_n) D \bar{\eta}_\mathrm{d}. \label{eq3 in proof 1}
		\end{align}
		
		
		According to (\ref{dot eta id n}) and (\ref{bar eta i n}), it is easy to obtain that  
		\begin{align}
			\bar{\eta}_{i\mathrm{d}}(t) &= \mathrm{e}^{A_{i\mathrm{d}}(t-t_k^i)} \eta_{i\mathrm{d}}(t_k^i), \nonumber \\
			\eta_{i\mathrm{d}}(t) &= \mathrm{e}^{(A_{i\mathrm{d}}-L_{i\mathrm{d}}H_{i\mathrm{d}})(t-t_k^i)} \eta_{i\mathrm{d}}(t_k^i), \nonumber
		\end{align}
		where $t\in [t_k^i, t_{k+1}^i)$. Therefore, it follows that 
		\begin{align}
			\bar{\eta}_{\mathrm{d}}(t) = G \eta_{\mathrm{d}}(t), \label{relation between eta and eta bar n}
		\end{align}
		where 
		\begin{align}
			G = \mathrm{diag}\{G_i \}_{i\in\mathcal{V}}, \,\, G_i = \mathrm{e}^{A_{i\mathrm{d}}(t-t_k^i)} \mathrm{e}^{(A_{i\mathrm{d}}-L_{i\mathrm{d}}H_{i\mathrm{d}})(t_k^i-t)}. \nonumber
		\end{align}
		The event-triggering condition (\ref{ET condition n}) implies that $t-t_k^i \in [0,\bar{\tau}_i)$, which ensures that every element of $G$ is finite.

		Combining (\ref{dot V})--(\ref{relation between eta and eta bar n}) yields
		\begin{align}
			\dot{V} & \leq -\mu_\mathrm{d} \eta_{\mathrm{d}}^\mathrm{T} \eta_{\mathrm{d}}  + \eta_\mathrm{d}^\mathrm{T} A_\mathrm{r}^\mathrm{T} A_\mathrm{r} \eta_\mathrm{d} + \eta_\mathrm{u}^\mathrm{T} \eta_\mathrm{u}  + 2 \eta_\mathrm{u}^\mathrm{T} A_\mathrm{u} \eta_\mathrm{u}  \nonumber \\
			&\quad - \frac{c}{2}\eta_\mathrm{u}^\mathrm{T} U^\mathrm{T} (\mathcal{L}\otimes I_n) U \eta_\mathrm{u} + \frac{c}{2} \sum_{i\in\mathcal{V}} \Big(\dot{\rho}_i + 4 l_{ii} \Vert \tilde{\eta}_{i\mathrm{u}} \Vert^2    \nonumber \\
			&\quad - \beta \sum_{j\in\mathcal{N}_i} a_{ij} \| \bar{\eta}_i - \bar{\eta}_j \|^2 \Big) + \frac{2-\beta}{1-\beta} \eta_\mathrm{d}^\mathrm{T} Q \eta_\mathrm{d}. \label{ineq5 in proof 1}
		\end{align}
		Note that $\tilde{\eta}_{i\mathrm{u}} = \tilde{x}_{i\mathrm{u}}$ and $\bar{\eta}_i-\bar{\eta}_j=\bar{x}_i-\bar{x}_j$, then substituting (\ref{dynamic variable n}) into (\ref{ineq5 in proof 1}) yields
		\begin{align}
			\dot{V} & \leq -\mu_\mathrm{d} \eta_{\mathrm{d}}^\mathrm{T} \eta_{\mathrm{d}}  + \eta_\mathrm{d}^\mathrm{T} A_\mathrm{r}^\mathrm{T} A_\mathrm{r} \eta_\mathrm{d} +  \eta_\mathrm{u}^\mathrm{T} \eta_\mathrm{u}  + 2 \eta_\mathrm{u}^\mathrm{T} A_\mathrm{u} \eta_\mathrm{u}  \nonumber \\
			&\quad - \frac{c}{2}\eta_\mathrm{u}^\mathrm{T} U^\mathrm{T} (\mathcal{L}\otimes I_n) U \eta_\mathrm{u} + \frac{2-\beta}{1-\beta} \eta_\mathrm{d}^\mathrm{T} Q \eta_\mathrm{d}    \nonumber \\
			&\quad  + \frac{c}{2} \sum_{i\in\mathcal{V}} \Big(-\delta_i \rho_i  + \gamma_i \|H_i \hat{x}_i-y_i\|^2 \Big). \nonumber 
		\end{align}
		Note that $H_i\hat{x}_i-y_i = H_i\eta_i =H_i T_i T_i^\mathrm{T} \eta_i =H_{i\mathrm{d}} \eta_{i\mathrm{d}} $. Then, it follows that
		\begin{align}
			\dot{V} & \leq -\mu_\mathrm{d} \eta_{\mathrm{d}}^\mathrm{T} \eta_{\mathrm{d}}  + \eta_\mathrm{d}^\mathrm{T} A_\mathrm{r}^\mathrm{T} A_\mathrm{r} \eta_\mathrm{d} + \eta_\mathrm{u}^\mathrm{T} \eta_\mathrm{u}  + 2 \eta_\mathrm{u}^\mathrm{T} A_\mathrm{u} \eta_\mathrm{u}  \nonumber \\
			&\quad - \frac{c}{2}\eta_\mathrm{u}^\mathrm{T} U^\mathrm{T} (\mathcal{L}\otimes I_n) U \eta_\mathrm{u} + \frac{2-\beta}{1-\beta} \eta_\mathrm{d}^\mathrm{T} Q \eta_\mathrm{d}    \nonumber \\
			&\quad  + \frac{c}{2} \sum_{i\in\mathcal{V}} \Big(-\delta_i \rho_i  + \gamma_i \|H_{i\mathrm{d}} \eta_{i\mathrm{d}}\|^2 \Big)\nonumber \\
			&\leq -\mu_\mathrm{d} \eta_{\mathrm{d}}^\mathrm{T} \eta_{\mathrm{d}}  + \eta_\mathrm{d}^\mathrm{T} A_\mathrm{r}^\mathrm{T} A_\mathrm{r} \eta_\mathrm{d} + \eta_\mathrm{u}^\mathrm{T} \eta_\mathrm{u}  + 2 \eta_\mathrm{u}^\mathrm{T} A_\mathrm{u} \eta_\mathrm{u}  \nonumber \\
			&\quad - \frac{c}{2}\eta_\mathrm{u}^\mathrm{T} U^\mathrm{T} (\mathcal{L}\otimes I_n) U \eta_\mathrm{u} - \frac{c\underline{\delta}}{2} \sum_{i\in\mathcal{V}}  \rho_i     \nonumber \\
			&\quad +\frac{c\bar{\gamma}}{2} \eta_\mathrm{d}^\mathrm{T}H_\mathrm{d}^\mathrm{T} H_\mathrm{d}\eta_\mathrm{d}  + \frac{2-\beta}{1-\beta} \eta_\mathrm{d}^\mathrm{T} Q \eta_\mathrm{d},\label{ineq6 in proof 1}
		\end{align}
		where $\underline{\delta} = \mathrm{min}_{i\in\mathcal{V}} \{\delta_i\}$, $\bar{\gamma} = \mathrm{max}_{i\in\mathcal{V}} \{\gamma_i\}$, $Q = G^\mathrm{T} D^\mathrm{T} (\mathcal{L}\otimes I_n) D G$.

		Let $Q_\mathrm{M} = \max_{t-t_k^i\in[0,\bar{\tau}_i)} \|Q\|$, $0<\alpha<\frac{1}{\|P_\mathrm{d}\|}$, and choose $\mu_\mathrm{d}$ such that $\mu_\mathrm{d} \geq \frac{ \| A_\mathrm{r}^\mathrm{T} A_\mathrm{r}  \| + \frac{c\bar{\gamma}}{2}\|H_\mathrm{d}^\mathrm{T}H_\mathrm{d}\| + \frac{2-\beta}{1-\beta} Q_\mathrm{M}}{1-\alpha\|P_\mathrm{d}\|}$. Then it follows from (\ref{choice of c n}) and (\ref{ineq6 in proof 1}) that 
		\begin{align}
			\dot{V}  &\leq -\alpha \mu_\mathrm{d} \eta_{\mathrm{d}}^\mathrm{T} P_\mathrm{d} \eta_{\mathrm{d}} - \eta_\mathrm{u}^\mathrm{T} \eta_\mathrm{u}  - \frac{c\underline{\delta}}{2} \sum_{i\in\mathcal{V}} \rho_i \leq -\xi V,  \nonumber
		\end{align}
		where $\xi = \min\{1,\alpha,\underline{\delta}\}$. According to the Comparison Principle \cite[Lemma 3.4]{khalil2002nonlinear}, we have $V \leq V(0) \mathrm{e}^{-\xi t}$ for all $t\geq 0$, which implies that both $\eta_\mathrm{d}$ and $\eta_\mathrm{u}$ converge to zero exponentially. The proof is completed.
	\end{proof}

	\subsection{Inter-Event Time Analysis}
	In the following, we demonstrate that the event-triggering mechanism (\ref{ET condition n}) ensures positive MIETs. Define the inter-event time as $\tau_k^i \triangleq t_{k+1}^i - t_k^i$ and we give the following theorem.
	\begin{theorem}
		Under the conditions in Theorem \ref{theorem 3.1}, it holds that  $\inf_{k\in\mathbb{N}} \{ \tau_k^i \} > 0$ for all $i\in\mathcal{V}$.
	\end{theorem}
	
	\begin{proof}
		Motivated by \cite{zhan2025dynamic}, we define the following two comparison functions:
		\begin{align}
			\bar{\psi}_i = \frac{4l_{ii} \| \tilde{x}_{i\mathrm{u}} \|^2  }{\beta \sum_{j \in \mathcal{N}_i} a_{ij} \| \bar{x}_i -  \bar{x}_j \|^2 + \kappa_i \rho_i}, \quad
			\psi_i = \frac{ \| \tilde{x}_{i} \|^2  }{ \rho_i} . \nonumber
		\end{align}
		Note that at the event instant $t_k^i$, both $\bar{\psi}_i$ and $\psi_i$ equal zero. According to the node-based event-triggering mechanism (\ref{ET condition n})-(\ref{dynamic variable n}), the inter-event time $\tau_k^i$ is the time that it takes for the function $\bar{\psi}_i$ to evolve from $0$ to $1$ for the first time after $t_k^i$. Moreover, it is easy to obtain that
		\begin{align}
			\bar{\psi}_i \leq \frac{4l_{ii}\| \tilde{x}_{i\mathrm{u}} \|^2 }{\kappa_i \rho_i}. \label{relation between bar psi and psi}
		\end{align}
		Since matrices $U_i$ and $D_i$ are the column blocks of the orthonormal matrix $T_i$, we have
		\begin{align}
			\| \tilde{x}_i  \|^2 &= \| U_i \tilde{x}_{i\mathrm{u}} + D_i \tilde{x}_{i\mathrm{d}}   \|^2 = \| \tilde{x}_{i\mathrm{u}} \|^2 + \Vert \tilde{x}_{i\mathrm{d}} \Vert^2, \label{tilde x i square}
		\end{align}
		which implies that $\| \tilde{x}_{i\mathrm{u}} \|^2 \leq \| \tilde{x}_i  \|^2$. Therefore, it follows from (\ref{relation between bar psi and psi}) that 
		\begin{align}
			\bar{\psi}_i  \leq \frac{4l_{ii}\| \tilde{x}_{i} \|^2 }{\kappa_i \rho_i}  = \frac{4l_{ii}}{\kappa_i} \psi_i, \label{bar psi i ineq}
		\end{align}
		which implies that the inter-event time $\tau_k^i$ is lower bounded by the time required for $\psi_i$ to transfer from $0$ to $\frac{\kappa_i}{4 l_{ii}}$.
		
		Taking the time derivative of $\psi_i$ along the solution of $\tilde{x}_i$ and $\rho_i$ yields
		\begin{align}
			\dot{\psi}_i &= \frac{2 \tilde{x}_i^\mathrm{T} \dot{\tilde{x}}_i}{\rho_i} - \frac{\| \tilde{x}_i \|^2}{\rho_i^2} \dot{\rho}_i \nonumber \\
			&= \frac{2 \tilde{x}_i^\mathrm{T} }{\rho_i} \Big[A\tilde{x}_i + L_i H_i \eta_i + cM_i \sum_{j\in\mathcal{N}_i} a_{ij}(\bar{x}_i - \bar{x}_j)\Big] \nonumber \\
			&\quad - \frac{\| \tilde{x}_i \|^2}{\rho_i^2} \Big(-\delta_i \rho_i - 4l_{ii} \| \tilde{x}_{i\mathrm{u}} \|^2     \nonumber \\
			&\quad +  \beta \sum_{j \in \mathcal{N}_i} a_{ij} \| \bar{x}_i -  \bar{x}_j \|^2  + \gamma_i \|H_i \hat{x}_i-y_i\|^2 \Big) \nonumber \\
			&\leq (\|A + A^\mathrm{T}\| + \delta_i) \frac{ \| \tilde{x}_{i} \|^2  }{ \rho_i} + 4l_{ii} \frac{ \| \tilde{x}_{i} \|^4  }{ \rho_i^2}   \nonumber \\
			&\quad + \frac{2\tilde{x}_i^T L_i H_i \eta_i}{\rho_i} - \frac{\gamma_i \| \tilde{x}_i \|^2}{\rho_i^2} \eta_i^\mathrm{T}H_i^\mathrm{T}H_i \eta_i  \nonumber \\
			&\quad + \frac{2c\tilde{x}_i^\mathrm{T} M_i \sum_{j\in\mathcal{N}_i}a_{ij}(\bar{x}_i - \bar{x}_j)}{\rho_i} \nonumber \\
			&\quad - \frac{\beta \| \tilde{x}_{i} \|^2  }{ \rho_i^2} \sum_{j\in\mathcal{N}_i} a_{ij}\| \bar{x}_i - \bar{x}_j \|^2.  \label{dot psi_i 1}
		\end{align}
		By Young's inequality, we have
		\begin{align}
			&\frac{2c\tilde{x}_i^\mathrm{T} M_i (\bar{x}_i - \bar{x}_j)}{\rho_i} - \frac{\beta \| \tilde{x}_{i} \|^2  }{ \rho_i^2} \| \bar{x}_i - \bar{x}_j \|^2 \leq \frac{\| M_i^\mathrm{T} M_i \|c^2}{\beta} \label{eq1 in proof 2}
		\end{align}
		and
		\begin{align}
			\frac{2\tilde{x}_i^\mathrm{T} L_i H_i \eta_i}{\rho_i} - \frac{\gamma_i \| \tilde{x}_i \|^2}{\rho_i^2} \eta_i^\mathrm{T}H_i^\mathrm{T}H_i \eta_i \leq \frac{\|L_i^\mathrm{T} L_i \|}{\gamma_i}. \label{dot psi_i component}
		\end{align}
		Substituting (\ref{eq1 in proof 2}) and (\ref{dot psi_i component}) into (\ref{dot psi_i 1}) yields
		\begin{align}
			\dot{\psi}_i &\leq  \frac{\| M_i^T M_i \| l_{ii} c^2}{\beta} + \frac{\|L_i^\mathrm{T} L_i \|}{\gamma_i}  + (\|A + A^\mathrm{T}\| + \delta_i) \psi_i \nonumber \\
			&\quad + 4l_{ii} \psi_i^2. \nonumber
		\end{align}
		Then by the Comparison Principle \cite[Lemma 3.4]{khalil2002nonlinear}, there holds $\psi_i(t) \leq \phi_i(t)$ for all $t\geq 0$, where $\phi_i$ is the solution to the differential equation
		\begin{align}
			\dot{\phi}_i = g_i(\phi_i), \quad t\in[t^i_k,t^i_{k+1}),  \nonumber
		\end{align}
		where $	 \phi_i(t_k^{i+}) = \psi_i(t_k^{i+}) = 0$ and $g_i(\phi_i)\triangleq \frac{\| M_i^\mathrm{T} M_i \|l_{ii} c^2}{\beta} + \frac{\|L_i^\mathrm{T} L_i \|}{\gamma_i} + (\|A + A^\mathrm{T}\| + \delta_i) \phi_i + 4l_{ii} \phi_i^2 $.
		
		Note that the time for $\phi_i$ to transfer from $0$ to $\frac{\kappa_i}{4l_{ii}}$ is 
		\begin{align}
			\underline{\tau}_i = \int_{0}^{\frac{\kappa_i}{4l_{ii}}} \frac{1}{g_i(s)} \mathrm{d} s. \label{theoretical lower bound of MIET n}
		\end{align}
		Obviously, $\underline{\tau}_i$ is a strictly positive constant. Given that $\psi_i(t) \leq \phi_i(t)$ and the inequality (\ref{bar psi i ineq}) holds, it follows that $\underline{\tau}_i$ is the lower bound of the time required by $\bar{\psi}_i$ to evolve from $0$ to $1$. Therefore, the inter-event times $\tau_k^i$ satisfy $\inf_{k\in\mathbb{N}}\{\tau^i_k\} \geq \underline{\tau}_i >0$ for all $i\in\mathcal{V}$, which completes the proof.
	\end{proof}
	
	\begin{remark}
		In \cite{sun2020event,liu2021resilient,su2024adaptive,long2025adaptive}, the event-triggered distributed observers are only proven to be able to exclude Zeno behavior of the resulting event times. In contrast, we prove that our method can ensure strictly positive MIETs, which lays a solid foundation for the practicality of the designed event-triggered distributed observer.
	\end{remark}
	
	\begin{remark}
		In \cite{zhu2025hybrid}, the event-triggered distributed observer is also proven to have strictly positive MIETs. The primary distinctions between our work and  \cite{zhu2025hybrid} lie in the following aspects: i) In \cite{zhu2025hybrid}, the positive MIETs are ensured by time-regularization. Specifically, the event-triggering mechanism imposes a constraint that the next triggering instant can only occur after a fixed time interval to ensure positive inter-event intervals, which, however, is technically different from our work where the existence of a positive lower bound of MIETs is proved to be an inherent property of the dynamic event-triggering mechanism (\ref{ET condition n}); ii) Only a node-based event-triggering mechanism is considered in \cite{zhu2025hybrid}. In contrast, we construct both node-based and edge-based event-triggering mechanisms (the edge-based case is presented in the next section); iii) In addition, numerical results (see Table \ref{comparative study}) show that the guaranteed MIET level provided by our method is much larger than that by \cite{zhu2025hybrid}.
	\end{remark}
	
	\section{Edge-based Event-Triggered Distributed Observer }\label{Edge-based Event-Triggered Distributed Observer}
	
	\subsection{Observer Design}
	
	For the edge-based case, we design a dynamic event-triggered distributed observer of the following form:
	\begin{align}
		\dot{\hat{x}}_i (t) &= A \hat{x}_i(t) + L_i (y_i(t) - H_i \hat{x}_i(t)) \nonumber \\
		&\quad - \check{c}M_i \sum_{j\in \mathcal{N}_i} a_{ij} (\bar{x}_{ij}(t)-\bar{x}_{ji}(t)), \label{observer e}
	\end{align}
	where $\bar{x}_{ij}(t) = \mathrm{e}^{A(t-t_k^{ij})} \hat{x}_i(t_k^{ij})$, $t \in \left[t_k^{ij},t_{k+1}^{ij}\right)$, with $\{t_k^{ij}\}_{k\in\mathbb{N}}$ being the sampling instant sequence for edge $(i,j)$. Moreover, the gain $\check{c}$ is any constant such that
	\begin{align}
		\check{c}\geq \frac{2 +\| A_\mathrm{u}+A_\mathrm{u}^\mathrm{T} \|}{(1-2\epsilon)\lambda_{\mathrm{min}}(U^\mathrm{T}(\mathcal{L}\otimes I_n) U)}, \label{choice of c e}
	\end{align}	
	where $\epsilon$ is any constant belonging to  $(0,\frac{1}{2})$ and similarly, the graph $\mathcal{G}$ is assumed to connected so that the right hand side of (\ref{choice of c e}) is a finite constant by Lemma \ref{lemma 2.1}. The design of $L_i$ and $M_i$ is the same as that in (\ref{gain matrix}).

	We now present the edge-based event-triggering mechanism. Define $\tilde{x}_{ij} = \bar{x}_{ij} - \hat{x}_i$ and perform the detectability decomposition on $\tilde{x}_{ij}$ using $T_i$ as $\tilde{x}_{ij} = T_i \begin{bmatrix}
		\tilde{x}_{ij\mathrm{d}} \\
		\tilde{x}_{ij\mathrm{u}}
	\end{bmatrix}$. The sampling instant sequence $\{t_k^{ij}\}_{k\in\mathbb{N}}$ is generated by the following dynamic event-triggering mechanism:
	\begin{align}
		&t_{k+1}^{ij} = \inf \{t>t_k^{ij}|f_{ij}(t) > 0 \lor t-t_k^{ij} \geq \bar{\tau}_{ij}\}, \quad t_0^{ij} = 0, \label{ET condition e}
	\end{align}
	where
	\begin{align}
		f_{ij}(t) &= 4a_{ij} \| \tilde{x}_{ij\mathrm{u}} \|^2  - \frac{\check{\beta}}{N_i} \sum_{k \in \mathcal{N}_i} a_{ik} \| \bar{x}_{ik} -  \bar{x}_{ki} \|^2 \nonumber \\
		&\quad - \kappa_{ij} \rho_{ij}(t), \label{ET function e}
	\end{align}
	and $\rho_{ij}(\cdot) : \left[0,\infty\right) \rightarrow \left(0,\infty\right)$ evolves according to
	\begin{align}
		\dot{\rho}_{ij} &= -\delta_{ij} \rho_{ij} - 4a_{ij} \| \tilde{x}_{ij\mathrm{u}} \|^2  +  \frac{\check{\beta}}{N_i} \sum_{k \in \mathcal{N}_i} a_{ik} \| \bar{x}_{ik} -  \bar{x}_{ki} \|^2 \nonumber \\
		&\quad      + \gamma_{ij} \| H_i\hat{x}_i-y_i \|^2. \label{dynamic variable e} 
	\end{align}
	In (\ref{ET condition e})--(\ref{dynamic variable e}), $\bar{\tau}_{ij}$, $\delta_{ij}$, $\kappa_{ij}$, $\gamma_{ij}$ and $\rho_{ij}(0)$ are any positive constants, and $0<\check{\beta}<1$. Similar to the dynamic variables $\rho_i(t)$ in (\ref{dynamic variable n}), one can easily obtain that $\rho_{ij}(t)>0$ for all $t\geq0$.

	\subsection{Convergence Analysis}
	In this section, we prove that the distributed observer (\ref{observer e}) together with the edge-based dynamic event-triggering mechanism (\ref{ET condition e}) guarantees exponential convergence of the estimation error. 
	
	\begin{theorem} \label{theorem 4.1}
		Consider the system (\ref{LTI system}) and (\ref{y}) over the communication network $\mathcal{G}$. Suppose that $\mathcal{G}$ is undirected and connected, and the pair $(A,H)$ is detectable. Then the observer (\ref{observer e})--(\ref{choice of c e}) with the edge-based dynamic event-triggering mechanism (\ref{ET condition e}) guarantees that the estimation error $\eta_i$ converges to zero exponentially.
	\end{theorem}
	
	\begin{proof}
		Similar to the derivations in (\ref{dot eta i n})--(\ref{dot eta iu n}), the dynamics of $\eta_{i\mathrm{d}}$ and $\eta_{i\mathrm{u}}$ can be obtained:
		\begin{align}
			\dot{\eta}_{i\mathrm{d}} &= (A_{i\mathrm{d}} - L_{i\mathrm{d}} H_{i\mathrm{d}} ) \eta_{i\mathrm{d}}, \label{dot eta id e} \\
			\dot{\eta}_{i\mathrm{u}} &= A_{i\mathrm{r}} \eta_{i\mathrm{d}} + A_{i\mathrm{u}} \eta_{i\mathrm{u}} - \check{c}U_i^\mathrm{T} \sum_{j \in \mathcal{N}_i} a_{ij}(\bar{\eta}_{ij} - \bar{\eta}_{ji}), \label{dot eta iu e}
		\end{align}
		where 
		\begin{align}
			\bar{\eta}_{ij} =  T_i\begin{bmatrix}
				\bar{\eta}_{ij\mathrm{d}} \\
				\bar{\eta}_{ij\mathrm{u}}
			\end{bmatrix}. \label{bar eta ij}
		\end{align} 
		
		Let $\tilde{\eta}_{ij} = \bar{\eta}_{ij} - \eta_i = T_i \begin{bmatrix}
			\tilde{\eta}_{ij\mathrm{d}} \\
			\tilde{\eta}_{ij\mathrm{u}}
		\end{bmatrix}$. It is easy to obtain that 
		\begin{equation}
			\tilde{\eta}_{ij}=\tilde{x}_{ij},\,\,\tilde{\eta}_{ij\mathrm{d}}=\tilde{x}_{ij\mathrm{d}},\,\,\tilde{\eta}_{ij\mathrm{u}}=\tilde{x}_{ij\mathrm{u}}. \nonumber
		\end{equation}
		
		
		We now analyze the stability of the error dynamics (\ref{dot eta id e}) and (\ref{dot eta iu e}). Choose a candidate Lyapunov function as 
		\begin{align}
			W =  \check{\mu}_\mathrm{d} \eta_{\mathrm{d}}^\mathrm{T} P_\mathrm{d} \eta_\mathrm{d} +   \eta_{\mathrm{u}}^\mathrm{T} \eta_{\mathrm{u}} + \frac{\check{c}}{2} \sum_{i\in\mathcal{V}} \sum_{j\in\mathcal{N}_i} \rho_{ij}, \label{Lyapunov function e}
		\end{align} 
		where $\eta_{\mathrm{d}}$ and $\eta_{\mathrm{u}}$ are defined as in (\ref{concatenated variables}), $P_{\mathrm{d}}$ is the positive-definite matrix satisfying (\ref{Riccati equation}), and  $\check{\mu}_\mathrm{d}$ is a positive constant to be determined later. 
		
		Along the solution to (\ref{dot eta id e}) and (\ref{dot eta iu e}) and by (\ref{Riccati equation}), the time derivative of $W$ is given by
		\begin{align}
			\dot{W} &= -\check{\mu}_\mathrm{d} \eta_{\mathrm{d}}^\mathrm{T} \eta_{\mathrm{d}}  + \sum_{i\in\mathcal{V}} 2\eta_{i\mathrm{u}}^\mathrm{T} \Big[ A_{i\mathrm{r}} \eta_{i\mathrm{d}} + A_{i\mathrm{u}} \eta_{i\mathrm{u}}   \nonumber \\
			&\quad - \check{c}U_i^\mathrm{T} \sum_{j \in \mathcal{N}_i} a_{ij} (\bar{\eta}_{ij}-\bar{\eta}_{ji})   \Big] + \frac{\check{c}}{2} \sum_{i\in\mathcal{V}} \sum_{j\in\mathcal{N}_i} \dot{\rho}_{ij}  \nonumber \\
			&= -\check{\mu}_\mathrm{d} \eta_{\mathrm{d}}^\mathrm{T} \eta_{\mathrm{d}} + 2 \eta_\mathrm{u}^\mathrm{T} A_\mathrm{r} \eta_\mathrm{d} + 2 \eta_\mathrm{u}^\mathrm{T} A_\mathrm{u} \eta_\mathrm{u}  \nonumber \\
			&\quad  - 2\check{c} \sum_{i\in\mathcal{V}} \eta_{i\mathrm{u}}^\mathrm{T} U_i^\mathrm{T} \sum_{j \in \mathcal{N}_i} a_{ij} (\bar{\eta}_{ij}-\bar{\eta}_{ji}) + \frac{\check{c}}{2} \sum_{i\in\mathcal{V}} \sum_{j\in\mathcal{N}_i} \dot{\rho}_{ij}. \label{dot V e 1}
		\end{align}
		
		Note that 
		\begin{align}
			&-2\check{c}\sum_{i\in\mathcal{V}} \eta_{i\mathrm{u}}^\mathrm{T} U_i^\mathrm{T} \sum_{j\in\mathcal{N}_i} a_{ij} (\bar{\eta}_{ij}-\bar{\eta}_{ji}) \nonumber \\
			&= -\check{c}\sum_{i\in\mathcal{V}} (\eta_{i\mathrm{u}}^\mathrm{T} U_i^\mathrm{T} - \eta_{j\mathrm{u}}^\mathrm{T} U_j^\mathrm{T}) \sum_{j\in\mathcal{N}_i} a_{ij} (\bar{\eta}_{ij}-\bar{\eta}_{ji}) \nonumber \\
			&= -\check{c}\sum_{i\in\mathcal{V}} (\eta_{i\mathrm{u}}^\mathrm{T} U_i^\mathrm{T} - \eta_{j\mathrm{u}}^\mathrm{T} U_j^\mathrm{T}) \sum_{j\in\mathcal{N}_i} a_{ij} (U_i \bar{\eta}_{ij\mathrm{u}}- U_j \bar{\eta}_{ji\mathrm{u}}) \nonumber \\
			&\quad -\check{c} \sum_{i\in\mathcal{V}} (\eta_{i\mathrm{u}}^\mathrm{T} U_i^\mathrm{T} - \eta_{j\mathrm{u}}^\mathrm{T} U_j^\mathrm{T}) \sum_{j\in\mathcal{N}_i} a_{ij} (D_i \bar{\eta}_{ij\mathrm{d}}- D_j \bar{\eta}_{ji\mathrm{d}}). \label{eq11}
		\end{align}
		Since $\tilde{\eta}_{ij\mathrm{u}} = \bar{\eta}_{ij\mathrm{u}} - \eta_{i\mathrm{u}}$, we have
		\begin{align}
			&\sum_{i\in\mathcal{V}} (\eta_{i\mathrm{u}}^\mathrm{T} U_i^\mathrm{T} - \eta_{j\mathrm{u}}^\mathrm{T} U_j^\mathrm{T}) \sum_{j\in\mathcal{N}_i} a_{ij} (U_i \bar{\eta}_{ij\mathrm{u}}- U_j \bar{\eta}_{ji\mathrm{u}}) \nonumber \\
			&= \sum_{i\in\mathcal{V}} (\eta_{i\mathrm{u}}^\mathrm{T} U_i^\mathrm{T} - \eta_{j\mathrm{u}}^\mathrm{T} U_j^\mathrm{T}) \sum_{j\in\mathcal{N}_i} a_{ij} (U_i \eta_{i\mathrm{u}}- U_j \eta_{j\mathrm{u}}) \nonumber \\
			&\quad + \sum_{i\in\mathcal{V}} (\eta_{i\mathrm{u}}^\mathrm{T} U_i^\mathrm{T} - \eta_{j\mathrm{u}}^\mathrm{T} U_j^\mathrm{T}) \sum_{j\in\mathcal{N}_i} a_{ij} (U_i \tilde{\eta}_{ij\mathrm{u}}- U_j \tilde{\eta}_{ji\mathrm{u}}) \nonumber \\
			&= \sum_{i\in\mathcal{V}} (\eta_{i\mathrm{u}}^\mathrm{T} U_i^\mathrm{T} - \eta_{j\mathrm{u}}^\mathrm{T} U_j^\mathrm{T}) \sum_{j\in\mathcal{N}_i} a_{ij} (U_i \eta_{i\mathrm{u}}- U_j \eta_{j\mathrm{u}}) \nonumber \\
			&\quad + \sum_{i\in\mathcal{V}} (\bar{\eta}_{ij\mathrm{u}}^\mathrm{T} U_i^\mathrm{T} - \bar{\eta}_{ji\mathrm{u}}^\mathrm{T} U_j^\mathrm{T} ) \sum_{j\in\mathcal{N}_i} a_{ij} (U_i \tilde{\eta}_{ij\mathrm{u}}- U_j \tilde{\eta}_{ji\mathrm{u}}) \nonumber \\
			&\quad - \sum_{i\in\mathcal{V}} ( \tilde{\eta}_{ij\mathrm{u}}^\mathrm{T} U_i^\mathrm{T}-  \tilde{\eta}_{ji\mathrm{u}}^\mathrm{T}U_j^\mathrm{T}) \sum_{j\in\mathcal{N}_i} a_{ij}(U_i \tilde{\eta}_{ij\mathrm{u}}- U_j \tilde{\eta}_{ji\mathrm{u}}) \nonumber \\
			&= \sum_{i\in\mathcal{V}}  \sum_{j\in\mathcal{N}_i} a_{ij} \| U_i \eta_{i\mathrm{u}}- U_j \eta_{j\mathrm{u}} \|^2 \nonumber \\
			&\quad + \sum_{i\in\mathcal{V}} (\bar{\eta}_{ij\mathrm{u}}^\mathrm{T} U_i^\mathrm{T} - \bar{\eta}_{ji\mathrm{u}}^\mathrm{T} U_j^\mathrm{T} ) \sum_{j\in\mathcal{N}_i} a_{ij} (U_i \bar{\eta}_{ij\mathrm{u}}- U_j \bar{\eta}_{ji\mathrm{u}}) \nonumber \\
			&\quad - \sum_{i\in\mathcal{V}} (\bar{\eta}_{ij\mathrm{u}}^\mathrm{T} U_i^\mathrm{T} - \bar{\eta}_{ji\mathrm{u}}^\mathrm{T} U_j^\mathrm{T} ) \sum_{j\in\mathcal{N}_i} a_{ij} (U_i \eta_{i\mathrm{u}}- U_j \eta_{j\mathrm{u}}) \nonumber \\
			&\quad - \sum_{i\in\mathcal{V}}  \sum_{j\in\mathcal{N}_i} a_{ij}\| U_i \tilde{\eta}_{ij\mathrm{u}}- U_j \tilde{\eta}_{ji\mathrm{u}} \|^2
		\end{align}
		which implies that 
		\begin{align}
			&\sum_{i\in\mathcal{V}} (\eta_{i\mathrm{u}}^\mathrm{T} U_i^\mathrm{T} - \eta_{j\mathrm{u}}^\mathrm{T} U_j^\mathrm{T}) \sum_{j\in\mathcal{N}_i} a_{ij} (U_i \bar{\eta}_{ij\mathrm{u}}- U_j \bar{\eta}_{ji\mathrm{u}}) \nonumber \\
			&= \frac{1}{2} \sum_{i\in\mathcal{V}}  \sum_{j\in\mathcal{N}_i} a_{ij} \| U_i \eta_{i\mathrm{u}}- U_j \eta_{j\mathrm{u}} \|^2 \nonumber \\
			&\quad + \frac{1}{2} \sum_{i\in\mathcal{V}}  \sum_{j\in\mathcal{N}_i} a_{ij} \| U_i \bar{\eta}_{ij\mathrm{u}}- U_j \bar{\eta}_{ji\mathrm{u}} \|^2 \nonumber \\
			&\quad - \frac{1}{2} \sum_{i\in\mathcal{V}}  \sum_{j\in\mathcal{N}_i} a_{ij}\| U_i \tilde{\eta}_{ij\mathrm{u}}- U_j \tilde{\eta}_{ji\mathrm{u}} \|^2. \label{eq22}
		\end{align}
		
		Applying Young's inequality and substituting (\ref{eq22}) into (\ref{eq11}), it can be obtained that
		\begin{align}
			&-2\check{c}\sum_{i\in\mathcal{V}} \eta_{i\mathrm{u}}^\mathrm{T} U_i^\mathrm{T} \sum_{j\in\mathcal{N}_i} a_{ij} (\bar{\eta}_{ij}-\bar{\eta}_{ji}) \nonumber \\
			&\leq -\check{c}\left(\frac{1}{2}-\epsilon\right) \sum_{i\in\mathcal{V}}  \sum_{j\in\mathcal{N}_i} a_{ij} \| U_i \eta_{i\mathrm{u}}- U_j \eta_{j\mathrm{u}} \|^2 \nonumber \\
			&\quad -\frac{\check{c}}{2} \sum_{i\in\mathcal{V}}  \sum_{j\in\mathcal{N}_i} a_{ij} \| U_i \bar{\eta}_{ij\mathrm{u}}- U_j \bar{\eta}_{ji\mathrm{u}} \|^2 \nonumber \\
			&\quad + \frac{\check{c}}{2} \sum_{i\in\mathcal{V}}  \sum_{j\in\mathcal{N}_i} a_{ij}\| U_i \tilde{\eta}_{ij\mathrm{u}}- U_j \tilde{\eta}_{ji\mathrm{u}} \|^2 \nonumber \\
			&\quad + \frac{\check{c}}{\epsilon} \sum_{i\in\mathcal{V}}  \sum_{j\in\mathcal{N}_i} a_{ij} \| D_i \bar{\eta}_{ij\mathrm{d}} - D_j \bar{\eta}_{ji\mathrm{d}} \|^2 \nonumber \\
			&\leq -\check{c}\left(\frac{1}{2}-\epsilon\right) \sum_{i\in\mathcal{V}}  \sum_{j\in\mathcal{N}_i} a_{ij} \| U_i \eta_{i\mathrm{u}}- U_j \eta_{j\mathrm{u}} \|^2 \nonumber \\
			&\quad -\frac{\check{c}}{2} \sum_{i\in\mathcal{V}}  \sum_{j\in\mathcal{N}_i} a_{ij} \| U_i \bar{\eta}_{ij\mathrm{u}}- U_j \bar{\eta}_{ji\mathrm{u}} \|^2 \nonumber \\
			&\quad + 2\check{c} \sum_{i\in\mathcal{V}}  \sum_{j\in\mathcal{N}_i} a_{ij}\| \tilde{\eta}_{ij\mathrm{u}}\|^2 + \frac{2\check{c}}{\epsilon} \sum_{i\in\mathcal{V}}  \sum_{j\in\mathcal{N}_i} a_{ij} \| \bar{\eta}_{ij\mathrm{d}}\|^2, \label{eq33}
		\end{align}
		where $\epsilon$ is any constant belonging to $(0,\frac{1}{2})$.
		
		Note that $\bar{\eta}_{ij}=U_i\bar{\eta}_{ij\mathrm{u}} + D_i \bar{\eta}_{ij\mathrm{d}}$. It can be obtained that 
		\begin{align}
			&-\frac{\check{c}}{2} \sum_{i\in\mathcal{V}}  \sum_{j\in\mathcal{N}_i} a_{ij} \| U_i \bar{\eta}_{ij\mathrm{u}}- U_j \bar{\eta}_{ji\mathrm{u}} \|^2 \nonumber \\
			&= -\frac{\check{c}}{2} \sum_{i\in\mathcal{V}}  \sum_{j\in\mathcal{N}_i} a_{ij} \| \bar{\eta}_{ij} - \bar{\eta}_{ji} \|^2 \nonumber \\
			&\quad -\frac{\check{c}}{2} \sum_{i\in\mathcal{V}} \sum_{j\in\mathcal{N}_i} a_{ij} \| D_i \bar{\eta}_{ij\mathrm{d}}- D_j \bar{\eta}_{ji\mathrm{d}} \|^2 \nonumber \\
			&\quad +\check{c} \sum_{i\in\mathcal{V}}  \sum_{j\in\mathcal{N}_i} a_{ij} (\bar{\eta}_{ij}^\mathrm{T} - \bar{\eta}_{ji}^\mathrm{T}) (D_i \bar{\eta}_{ij\mathrm{d}}- D_j \bar{\eta}_{ji\mathrm{d}}) \nonumber \\
			&\leq -\frac{\check{c}}{2} \sum_{i\in\mathcal{V}}  \sum_{j\in\mathcal{N}_i} a_{ij}  \| \bar{\eta}_{ij} - \bar{\eta}_{ji} \|^2 \nonumber \\
			&\quad -\frac{\check{c}}{2} \sum_{i\in\mathcal{V}}  \sum_{j\in\mathcal{N}_i} a_{ij} \| D_i \bar{\eta}_{ij\mathrm{d}}- D_j \bar{\eta}_{ji\mathrm{d}} \|^2 \nonumber \\
			&\quad +\frac{\check{c}(1-\check{\beta})}{2} \sum_{i\in\mathcal{V}}  \sum_{j\in\mathcal{N}_i} a_{ij} \| \bar{\eta}_{ij} - \bar{\eta}_{ji} \|^2 \nonumber \\
			&\quad + \frac{\check{c}}{2(1-\check{\beta})} \sum_{i\in\mathcal{V}}  \sum_{j\in\mathcal{N}_i} a_{ij} \|D_i \bar{\eta}_{ij\mathrm{d}}- D_j \bar{\eta}_{ji\mathrm{d}} \|^2 \nonumber \\
			&= -\frac{\check{c}\check{\beta}}{2} \sum_{i\in\mathcal{V}}  \sum_{j\in\mathcal{N}_i} a_{ij} \| \bar{\eta}_{ij} - \bar{\eta}_{ji} \|^2 \nonumber \\
			&\quad +\frac{\check{c}\check{\beta}}{2(1-\check{\beta})} \sum_{i\in\mathcal{V}}  \sum_{j\in\mathcal{N}_i} a_{ij} \|D_i \bar{\eta}_{ij\mathrm{d}}- D_j \bar{\eta}_{ji\mathrm{d}} \|^2 \nonumber \\
			&\leq -\frac{\check{c}\check{\beta}}{2} \sum_{i\in\mathcal{V}}  \sum_{j\in\mathcal{N}_i} a_{ij} \| \bar{\eta}_{ij} - \bar{\eta}_{ji} \|^2 \nonumber \\
			&\quad +\frac{2\check{c}\check{\beta}}{1-\check{\beta}} \sum_{i\in\mathcal{V}}  \sum_{j\in\mathcal{N}_i} a_{ij} \| \bar{\eta}_{ij\mathrm{d}}\|^2, \label{eq44}
		\end{align}
		where $0<\check{\beta}<1$, and the first inequality is due to Young's inequality.
		
		Combining (\ref{dot V e 1}), (\ref{eq33}) and (\ref{eq44}) yields
		\begin{align}
			\dot{W} &\leq -\check{\mu}_\mathrm{d} \eta_{\mathrm{d}}^\mathrm{T} \eta_{\mathrm{d}} + 2 \eta_\mathrm{u}^\mathrm{T} A_\mathrm{r} \eta_\mathrm{d} + 2 \eta_\mathrm{u}^\mathrm{T} A_\mathrm{u} \eta_\mathrm{u}  \nonumber \\
			&\quad  -\check{c}\left(\frac{1}{2}-\epsilon\right) \sum_{i\in\mathcal{V}}  \sum_{j\in\mathcal{N}_i} a_{ij} \| U_i \eta_{i\mathrm{u}}- U_j \eta_{j\mathrm{u}} \|^2 \nonumber \\
			&\quad -\frac{\check{c}\check{\beta}}{2} \sum_{i\in\mathcal{V}}  \sum_{j\in\mathcal{N}_i} a_{ij} \| \bar{\eta}_{ij} - \bar{\eta}_{ji} \|^2 + 2\check{c} \sum_{i\in\mathcal{V}}  \sum_{j\in\mathcal{N}_i} a_{ij}\| \tilde{\eta}_{ij\mathrm{u}}\|^2 \nonumber \\
			&\quad  + 2\check{c}\zeta  \sum_{i\in\mathcal{V}}  \sum_{j\in\mathcal{N}_i} a_{ij} \| \bar{\eta}_{ij\mathrm{d}}\|^2 + \frac{\check{c}}{2} \sum_{i\in\mathcal{V}} \sum_{j\in\mathcal{N}_i} \dot{\rho}_{ij}, \label{dot V e 2}
		\end{align}
		where $\zeta = \frac{1}{\epsilon}+\frac{\check{\beta}}{1-\check{\beta}}$. Since $\bar{\eta}_{ij}-\bar{\eta}_{ji}=\bar{x}_{ij}-\bar{x}_{ji}$ and $\tilde{\eta}_{ij\mathrm{u}}=\tilde{x}_{ij\mathrm{u}}$, it follows from (\ref{dot V e 2}) that
		\begin{align}
			\dot{W} &\leq -\check{\mu}_\mathrm{d} \eta_{\mathrm{d}}^\mathrm{T} \eta_{\mathrm{d}} + 2 \eta_\mathrm{u}^\mathrm{T} A_\mathrm{r} \eta_\mathrm{d} + 2 \eta_\mathrm{u}^\mathrm{T} A_\mathrm{u} \eta_\mathrm{u}  \nonumber \\
			&\quad  -\check{c}\left(\frac{1}{2}-\epsilon\right) \sum_{i\in\mathcal{V}} \sum_{j\in\mathcal{N}_i} a_{ij} \| U_i \eta_{i\mathrm{u}}- U_j \eta_{j\mathrm{u}} \|^2 \nonumber \\
			&\quad -\frac{\check{c}\check{\beta}}{2} \sum_{i\in\mathcal{V}}  \sum_{j\in\mathcal{N}_i} a_{ij} \| \bar{x}_{ij} - \bar{x}_{ji} \|^2  \nonumber \\
			&\quad + 2\check{c} \sum_{i\in\mathcal{V}} \sum_{j\in\mathcal{N}_i} a_{ij}\| \tilde{x}_{ij\mathrm{u}}\|^2 + 2\check{c}\zeta  \sum_{i\in\mathcal{V}} \sum_{j\in\mathcal{N}_i} a_{ij} \| \bar{\eta}_{ij\mathrm{d}}\|^2 \nonumber \\
			&\quad + \frac{\check{c}}{2} \sum_{i\in\mathcal{V}} \sum_{j\in\mathcal{N}_i} \dot{\rho}_{ij}, \label{dot V e 3}
		\end{align}

		According to (\ref{dot eta id e}) and (\ref{bar eta ij}), it is easy to obtain that  
		\begin{align}
			\bar{\eta}_{ij\mathrm{d}}(t) &= \mathrm{e}^{A_{i\mathrm{d}}(t-t_k^{ij})} \eta_{i\mathrm{d}}(t_k^{ij}),  &t\in [t_k^{ij}, t_{k+1}^{ij}), \nonumber \\
			\eta_{i\mathrm{d}}(t) &= \mathrm{e}^{(A_{i\mathrm{d}}-L_{i\mathrm{d}}H_{i\mathrm{d}})(t-t_k^{ij})} \eta_{i\mathrm{d}}(t_k^{ij}), &t\in [t_k^{ij}, t_{k+1}^{ij}). \nonumber
		\end{align}
		Therefore, it follows that 
		\begin{align}
			\bar{\eta}_{ij\mathrm{d}}(t) = G_{ij} \eta_{i\mathrm{d}}(t), \label{relation between eta and eta bar e}
		\end{align}
		where $G_{ij} = \mathrm{e}^{A_{i\mathrm{d}}(t-t_k^{ij})} \mathrm{e}^{(A_{i\mathrm{d}}-L_{i\mathrm{d}}H_{i\mathrm{d}})(t_k^{ij}-t)}$.
		Notably, the event-triggering condition (\ref{ET condition e}) implies that $t-t_k^{ij} \in [0,\bar{\tau}_{ij})$, which ensures that every element of $G_{ij}$ is finite.
		
		Combing (\ref{dynamic variable e}), (\ref{dot V e 3}) and (\ref{relation between eta and eta bar e}) yields
		\begin{align}
			\dot{W} & \leq -\check{\mu}_\mathrm{d} \eta_{\mathrm{d}}^\mathrm{T} \eta_{\mathrm{d}} + 2 \eta_\mathrm{u}^\mathrm{T} A_\mathrm{r} \eta_\mathrm{d} + 2 \eta_\mathrm{u}^\mathrm{T} A_\mathrm{u} \eta_\mathrm{u}  \nonumber \\
			&\quad  -\check{c}\left(\frac{1}{2}-\epsilon\right) \sum_{i\in\mathcal{V}}  \sum_{j\in\mathcal{N}_i} a_{ij} \| U_i \eta_{i\mathrm{u}}- U_j \eta_{j\mathrm{u}} \|^2   \nonumber \\
			&\quad  + 2\check{c}\zeta  \sum_{i\in\mathcal{V}}  \sum_{j\in\mathcal{N}_i} a_{ij} \| \bar{\eta}_{ij\mathrm{d}}\|^2   - \frac{\check{c}}{2} \sum_{i\in\mathcal{V}}  \sum_{j\in\mathcal{N}_i} \delta_{ij} \rho_{ij}\nonumber \\
			&\quad + \frac{\check{c}}{2} \sum_{i\in\mathcal{V}}  \sum_{j\in\mathcal{N}_i} \gamma_{ij} \| H_i\hat{x}_i-y_i \|^2. \label{dot V e 4}
		\end{align}
		
		Note that $ H_i\hat{x}_i-y_i  = H_i \eta_i = H_i T_i T_i^\mathrm{T} \eta_i = H_{i\mathrm{d}} \eta_{i\mathrm{d}}$. It can be obtained that
		\begin{align}
			&\frac{\check{c}}{2} \sum_{i\in\mathcal{V}}  \sum_{j\in\mathcal{N}_i} \gamma_{ij} \| H_i\hat{x}_i-y_i \|^2 \nonumber \\
			&= \frac{\check{c}}{2} \sum_{i\in\mathcal{V}}  \sum_{j\in\mathcal{N}_i} \gamma_{ij} \| H_{i\mathrm{d}} \eta_{i\mathrm{d}} \|^2 \nonumber \\
			&\leq \frac{\check{c}}{2} \sum_{i\in\mathcal{V}} \bar{\gamma_i} N_i \| H_{i\mathrm{d}}^\mathrm{T} H_{i\mathrm{d}} \| \eta_{i\mathrm{d}}^\mathrm{T}  \eta_{i\mathrm{d}} \leq \frac{\check{c} \bar{\theta}}{2} \eta_\mathrm{d}^\mathrm{T} \eta_\mathrm{d}, \label{77}
		\end{align}
		where $\bar{\gamma}_i=\max_{j\in\mathcal{N}_i}\{\gamma_{ij}\}$ and $\bar{\theta}=\!\max_{i\in\mathcal{V}}\{\bar{\gamma_i} N_i  \| H_{i\mathrm{d}}^\mathrm{T} H_{i\mathrm{d}}\| \}$.
		
		According to (\ref{relation between eta and eta bar e}), we have
		\begin{align}
			&2\check{c}\zeta  \sum_{i\in\mathcal{V}}  \sum_{j\in\mathcal{N}_i} a_{ij} \| \bar{\eta}_{ij\mathrm{d}}\|^2 \nonumber \\
			&\leq 2\check{c}\zeta  \sum_{i\in\mathcal{V}}  \sum_{j\in\mathcal{N}_i} \| G_{ij}^\mathrm{T} G_{ij} \| \eta_{i\mathrm{d}}^\mathrm{T}  \eta_{i\mathrm{d}} \nonumber \\
			&\leq 2\check{c}\zeta  \sum_{i\in\mathcal{V}}  \bar{g}_i N_i \eta_{i\mathrm{d}}^\mathrm{T}  \eta_{i\mathrm{d}} \leq 2\check{c}\zeta\bar{\nu} \eta_\mathrm{d}^\mathrm{T} \eta_\mathrm{d}, \label{88}
		\end{align}
		where $\bar{g}_i= \max_{j\in\mathcal{N}_i}\{\| G_{ij}^\mathrm{T} G_{ij} \|\}$ and $\bar{\nu}=\max_{i\in\mathcal{V}}\{\bar{g}_i N_i\}$.
		
		Substitute (\ref{77}) and (\ref{88}) into (\ref{dot V e 4}), then we have
		\begin{align}
			\dot{W} & \leq -\check{\mu}_\mathrm{d} \eta_{\mathrm{d}}^\mathrm{T} \eta_{\mathrm{d}} + 2 \eta_\mathrm{u}^\mathrm{T} A_\mathrm{r} \eta_\mathrm{d} + 2 \eta_\mathrm{u}^\mathrm{T} A_\mathrm{u} \eta_\mathrm{u} \nonumber \\
			&\quad  -\check{c}\left(\frac{1}{2}-\epsilon\right) \sum_{i\in\mathcal{V}}  \sum_{j\in\mathcal{N}_i} a_{ij} \| U_i \eta_{i\mathrm{u}}- U_j \eta_{j\mathrm{u}} \|^2 \nonumber \\
			&\quad  + \left(2\check{c}\zeta\bar{\nu} + \frac{\check{c}\bar{\theta}}{2}\right) \eta_\mathrm{d}^\mathrm{T} \eta_\mathrm{d} - \frac{\check{c}}{2} \sum_{i\in\mathcal{V}}  \sum_{j\in\mathcal{N}_i} \delta_{ij} \rho_{ij} \nonumber \\
			& \leq -\check{\mu}_\mathrm{d} \eta_{\mathrm{d}}^\mathrm{T} \eta_{\mathrm{d}} + \eta_\mathrm{u}^\mathrm{T} \eta_\mathrm{u}  + \| A_\mathrm{r}^\mathrm{T} A_\mathrm{r} \| \eta_\mathrm{d}^\mathrm{T} \eta_\mathrm{d}  \nonumber \\
			&\quad + \| A_\mathrm{u}+A_\mathrm{u}^\mathrm{T} \| \eta_\mathrm{u}^\mathrm{T} \eta_\mathrm{u} - 2\check{c}\left(\frac{1}{2}-\epsilon\right) \eta_\mathrm{u}^\mathrm{T} U^\mathrm{T}(\mathcal{L}\otimes I_n) U \eta_\mathrm{u} \nonumber \\
			&\quad  + \left(2\check{c}\zeta\bar{\nu} + \frac{\check{c}\bar{\theta}}{2}\right) \eta_\mathrm{d}^\mathrm{T} \eta_\mathrm{d} - \frac{\check{c}}{2} \sum_{i\in\mathcal{V}}  \sum_{j\in\mathcal{N}_i} \delta_{ij} \rho_{ij}, \label{dot V e 5}
		\end{align}
		where the second inequality is due to Young's inequality.
		
		Let $0<\check{\alpha}<\frac{1}{\|P_\mathrm{d}\|}$ and $\underline{\delta}_e = \min_{i\in\mathcal{V}, j\in\mathcal{N}_i}\{\delta_{ij}\}$. Choose $\check{\mu}_\mathrm{d} \geq \frac{2\check{c}\zeta\bar{\nu} + \frac{\check{c}\bar{\theta}}{2} + \|A_\mathrm{r}^\mathrm{T} A_\mathrm{r}\|}{1-\check{\alpha}\|P_\mathrm{d}\|}$ and set $\check{c}$ according to (\ref{choice of c e}), then it follows from (\ref{dot V e 5}) that
		\begin{align}
			\dot{W} &\leq - \check{\alpha}\check{\mu}_\mathrm{d} \eta_\mathrm{d}^\mathrm{T} P_\mathrm{d} \eta_\mathrm{d} - \eta_{\mathrm{u}}^\mathrm{T} \eta_{\mathrm{u}} - \frac{\check{c} \underline{\delta}_e}{2} \sum_{i\in\mathcal{V}}  \sum_{j\in\mathcal{N}_i} \delta_{ij} \rho_{ij} \nonumber \\
			& \leq -\check{\xi} W, \nonumber
		\end{align}
		where $\check{\xi}=\min\{\check{\alpha}, 1, \underline{\delta}_e \}$. According to the Comparison Principle \cite[Lemma 3.4]{khalil2002nonlinear}, we have $W \leq W(0) \mathrm{e}^{-\check{\xi} t}$ for all $t\geq 0$, which implies that both $\eta_\mathrm{d}$ and $\eta_\mathrm{u}$ converge to zero exponentially. The proof is completed.
	\end{proof}
	
	\subsection{Inter-Event Time Analysis}
	In the following, we demonstrate that the edge-based event-triggering mechanism (\ref{ET condition e}) ensures positive MIETs. Define the inter-event time as $\tau_k^{ij} \triangleq t_{k+1}^{ij} - t_k^{ij}$.
	\begin{theorem}
		Under the conditions in Theorem \ref{theorem 4.1}, it holds that  $\inf_{k\in\mathbb{N}} \{ \tau_k^{ij} \} > 0$ for all $(i,j)\in\mathcal{E}$.
	\end{theorem}
	
	\begin{proof}
		Define the following two comparison functions:
		\begin{align}
			&\bar{\psi}_{ij} = \frac{4a_{ij} \| \tilde{x}_{ij\mathrm{u}} \|^2  }{\frac{\check{\beta}}{N_i} \sum_{k \in \mathcal{N}_i} a_{ik} \| \bar{x}_{ik} -  \bar{x}_{ki} \|^2 + \kappa_{ij} \rho_{ij}}, \nonumber \\
			&\psi_{ij} = \frac{ \| \tilde{x}_{ij} \|^2  }{ \rho_{ij}}. \nonumber
		\end{align}
		According to the edge-based event-triggering mechanism (\ref{ET condition e})-(\ref{dynamic variable e}), the inter-event time $\tau_k^{ij}$ is the time required for $\bar{\psi}_{ij}$ to transition from $0$ to $1$ for the first time after $t_k^{ij}$. It holds that
		\begin{align}
			\bar{\psi}_{ij} \leq \frac{4a_{ij}\| \tilde{x}_{ij\mathrm{u}} \|^2 }{\kappa_{ij} \rho_{ij}}. \label{relation between bar psi and psi e}
		\end{align}
		Similar to (\ref{tilde x i square}), we have 
		\begin{align}
			\| \tilde{x}_{ij}  \|^2 &= \| U_i \tilde{x}_{ij\mathrm{u}} + D_i \tilde{x}_{ij\mathrm{d}}   \|^2 = \| \tilde{x}_{ij\mathrm{u}} \|^2 + \Vert \tilde{x}_{ij\mathrm{d}} \Vert^2, \nonumber
		\end{align}
		which implies that $\| \tilde{x}_{ij\mathrm{u}} \|^2 \leq \| \tilde{x}_{ij}  \|^2$. Therefore, it follows from (\ref{relation between bar psi and psi e}) that 
		\begin{align}
			\bar{\psi}_{ij} \leq  \frac{4a_{ij}\| \tilde{x}_{ij} \|^2 }{\kappa_{ij} \rho_{ij}} = \frac{4a_{ij}}{\kappa_{ij}} \psi_{ij}, \label{bar psi ij ineq}
		\end{align}
		which implies that the inter-event time $\tau_k^{ij}$ is lower bounded by the time required for $\psi_{ij}$ to transfer from $0$ to $\frac{\kappa_{ij}}{4 a_{ij}}$.	Then similar to the derivations as in (\ref{dot psi_i 1})--(\ref{dot psi_i component}), it can be shown that
		\begin{align}
			\dot{\psi}_{ij} &= \frac{2 \tilde{x}_{ij}^\mathrm{T} \dot{\tilde{x}}_{ij}}{\rho_{ij}} - \frac{\| \tilde{x}_{ij} \|^2}{\rho_{ij}^2} \dot{\rho}_{ij} \nonumber \\
			&= \frac{2 \tilde{x}_{ij}^\mathrm{T} }{\rho_{ij}} \Big[A\tilde{x}_{ij} + L_i H_i \eta_i + \check{c}M_i \sum_{k\in\mathcal{N}_i} a_{ik}(\bar{x}_{ik} - \bar{x}_{ki})\Big] \nonumber \\
			&\quad - \frac{\| \tilde{x}_{ij} \|^2}{\rho_{ij}^2} \Big(-\delta_{ij} \rho_{ij}(t) - 4a_{ij} \| \tilde{x}_{ij\mathrm{u}}(t) \|^2   \nonumber \\
			&\quad     +  \frac{\check{\beta}}{N_i} \sum_{k \in \mathcal{N}_i} a_{ik} \| \bar{x}_{ik}(t) -  \bar{x}_{ki}(t) \|^2 + \gamma_{ij} \| H_i\hat{x}_i-y_i \|^2 \Big) \nonumber \\
			&\leq (\|A + A^\mathrm{T}\| + \delta_{ij}) \frac{ \| \tilde{x}_{ij} \|^2  }{ \rho_{ij}} + 4a_{ij} \frac{ \| \tilde{x}_{ij} \|^4  }{ \rho_{ij}^2} \nonumber \\
			&\quad  + \frac{2\tilde{x}_{ij}^\mathrm{T} L_i H_i \eta_i}{\rho_{ij}}  - \frac{\gamma_{ij} \| \tilde{x}_{ij} \|^2}{\rho_{ij}^2} \eta_i^\mathrm{T}H_i^\mathrm{T}H_i \eta_i  \nonumber \\
			&\quad + \frac{2\tilde{c}\tilde{x}_{ij}^\mathrm{T} M_i \sum_{k\in\mathcal{N}_i}a_{ik}(\bar{x}_{ik} - \bar{x}_{ki})}{\rho_{ij}} \nonumber \\
			&\quad - \frac{\check{\beta} \| \tilde{x}_{ij} \|^2  }{ N_i \rho_{ij}^2} \sum_{k\in\mathcal{N}_i} a_{ik}\| \bar{x}_{ik} - \bar{x}_{ki} \|^2,  \label{dot psi ij}
		\end{align}
		By Young's inequality, we have
		\begin{align}
			&\frac{2\check{c}\tilde{x}_{ij}^\mathrm{T} M_i (\bar{x}_{ik} - \bar{x}_{ki})}{\rho_{ij}} - \frac{\check{\beta} \| \tilde{x}_{ij} \|^2  }{ N_i \rho_{ij}^2} \| \bar{x}_{ik} - \bar{x}_{ki} \|^2  \nonumber \\
			&\leq \frac{\|M_i^\mathrm{T} M_i\|N_i \check{c}^2}{\check{\beta}} \label{dot psi ij component 1}
		\end{align}
		and
		\begin{align}
			\frac{2\tilde{x}_{ij}^\mathrm{T} L_i H_i \eta_i}{\rho_{ij}} - \frac{\gamma_{ij} \| \tilde{x}_{ij} \|^2}{\rho_{ij}^2} \eta_i^\mathrm{T}H_i^\mathrm{T}H_i \eta_i \leq \frac{\|L_i^\mathrm{T}L_i\|}{\gamma_{ij}}. \label{dot psi ij component 2}
		\end{align}
		Substitute (\ref{dot psi ij component 1}) and (\ref{dot psi ij component 2}) into (\ref{dot psi ij}), then we have
		\begin{align}
			\dot{\psi}_{ij} &\leq \frac{\|M_i^\mathrm{T} M_i\|N_i l_{ii} \check{c}^2}{\check{\beta}} + \frac{\|L_i^\mathrm{T}L_i\|}{\gamma_{ij}}  \nonumber \\
			&\quad\, + (\|A + A^\mathrm{T}\| + \delta_{ij})\psi_{ij} + 4a_{ij}\psi_{ij}^2.
		\end{align}
		
		According to the Comparison Principle \cite[Lemma 3.4]{khalil2002nonlinear}, it holds that $\psi_{ij}(t) \leq \phi_{ij}(t)$ for all $t\geq 0$, where $\phi_{ij}$ is the solution of the following equation:
		\begin{align}
			\dot{\phi}_{ij} = g_{ij}(\phi_{ij}), \quad t\in[t_k^{ij},t_{k+1}^{ij}), \nonumber
		\end{align}
		where $\phi_{ij}(t_k^{ij+}) = \psi_{ij}(t_k^{ij+}) = 0$ and $g_{ij}(\phi_{ij}) \triangleq \frac{\|M_i^\mathrm{T} M_i\|N_i l_{ii} \check{c}^2}{\check{\beta}} + \frac{\|L_i^\mathrm{T}L_i\|}{\gamma_{ij}} + (\|A + A^\mathrm{T}\| + \delta_{ij})\phi_{ij}  + 4a_{ij}\phi_{ij}^2$.
		
		Note that the time for $\phi_{ij}$ to transfer from $0$ to $\frac{\kappa_{ij}}{4a_{ij}}$ is 
		\begin{align}
			\underline{\tau}_{ij} = \int_{0}^{\frac{\kappa_{ij}}{4a_{ij}}} \frac{1}{g_{ij}(s)} \mathrm{d} s. \nonumber
		\end{align}
		Obviously, $\underline{\tau}_{ij}$ is a positive constant. Since $\psi_{ij}(t) \leq \phi_{ij}(t)$ and the inequality (\ref{bar psi ij ineq}) holds, $\underline{\tau}_{ij}$ is the lower bound of the time required by $\bar{\psi}_{ij}$ to evolve from $0$ to $1$. Hence, the inter-event times $\tau_k^{ij}$ satisfy $\inf_{k\in\mathbb{N}}\{\tau^{ij}_k\} \geq \underline{\tau}_{ij} >0$ for all $(i,j)\in\mathcal{E}$, which completes the proof.
	\end{proof}
	
	\begin{remark}
		In \cite{li2025distributed}, a distributed observer with an edge-based switching event-triggering mechanism is proposed. The event-triggering mechanism switches between a periodic event-triggering mechanism and a continuous one. The primary distinctions between our method and the one in \cite{li2025distributed} lie in the following aspects: i) The observer design in \cite{li2025distributed} requires solving a LMI whose feasibility is not guaranteed. Specifically, a necessary condition for the feasibility of the LMI is the stability of the undetectable subsystem, i.e., $A_{i\mathrm{u}}$ is Hurwitz for $i\in \mathcal{V}$. In contrast, our method imposes no such constraint, thus offering broader applicability; ii) In \cite{li2025distributed}, edges $(i,j)$ and $(j,i)$ are triggered synchronously, while in our approach, all edges are triggered asynchronously; iii) Only asymptotic convergence of the estimation error can be ensured by the design in \cite{li2025distributed}, while we establish exponential convergence of the estimation error in this work; iv) In \cite{li2025distributed}, positive MIETs are guaranteed by a period that is additionally embedded into the triggering process, while in our work, the existence of positive MIET is an inherent property of the designed dynamic event-triggering mechanisms.
	\end{remark}
	\section{Numerical Examples} \label{Numerical Examples}
	
	\begin{figure}[htb]
		\centering
		\includegraphics[width=1.2in]{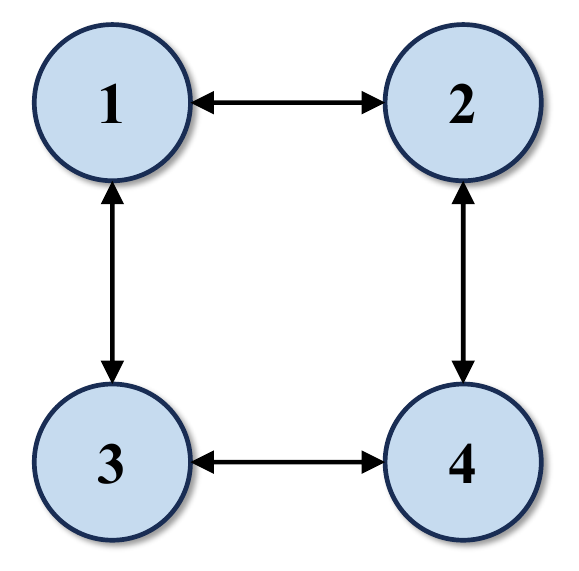}
		\caption{The communication graph among agents.}
		\label{fig1}
	\end{figure}

	In this section, we use numerical examples from \cite{kim2019completely} to illustrate the effectiveness of our approach. Specifically, a three-inertia system is considered, which is observed by four agents. The communication topology among agents is represented by an undirected graph, as illustrated in Fig. \ref{fig1}. The system dynamics and measurements are modeled by (\ref{LTI system}) and (\ref{y}) with
	\begin{align}
		A &= \begin{bmatrix}
			0 & 1 & 0 & 0 & 0 & 0 \\
			-\frac{k}{J}& 0 & \frac{k}{J} & 0 & 0 & 0  \\
			0 & 0 & 0 & 1 & 0 & 0  \\
			\frac{k}{J} & 0 & -\frac{2k}{J}& 0 & \frac{k}{J} & 0 \\
			0 & 0 & 0 & 0 & 0 & 1 \\
			0 & 0 & \frac{k}{J} & 0 & -\frac{k}{J} & 0
		\end{bmatrix},  \label{A sim} \\
		H &= \begin{bmatrix}
			H_1 \\
			H_2 \\
			H_3 \\
			H_4
		\end{bmatrix} = \begin{bmatrix}
			0 & 0 & 1 & 0 & 0 & 0 \\
			1 & 0 & -1& 0 & 0 & 0\\
			0 & 0 & 1 & 0 & -1& 0\\
			1 & 0 & 0 & 0 & -1& 0
		\end{bmatrix}, \label{H sim}
	\end{align}
	where $k$ denotes the torsional stiffness, $J$ denotes the moment of inertia, and $\frac{k}{J}$ is set to $1$ in the simulations. Besides, it can be seen that none of the pairs $(A,H_i)$ is detectable while the pair $(A,H)$ is detectable. 
	
	For the simulation purpose, the initial states of the system and local observers are selected randomly. We choose $L_{i\mathrm{d}}$ to place the eigenvalues of $A_{i\mathrm{d}} - L_{i\mathrm{d}} H_{i\mathrm{d}}$ at $\{ -1,-1,-1,-1 \}$ for $i=1,2,3$ and $\{-1,-1\}$ for $i=4$. In the following, we present three examples to verify the effectiveness of our method.
	
	\begin{figure}[!h]
	\centering
	\includegraphics[ width=\columnwidth]{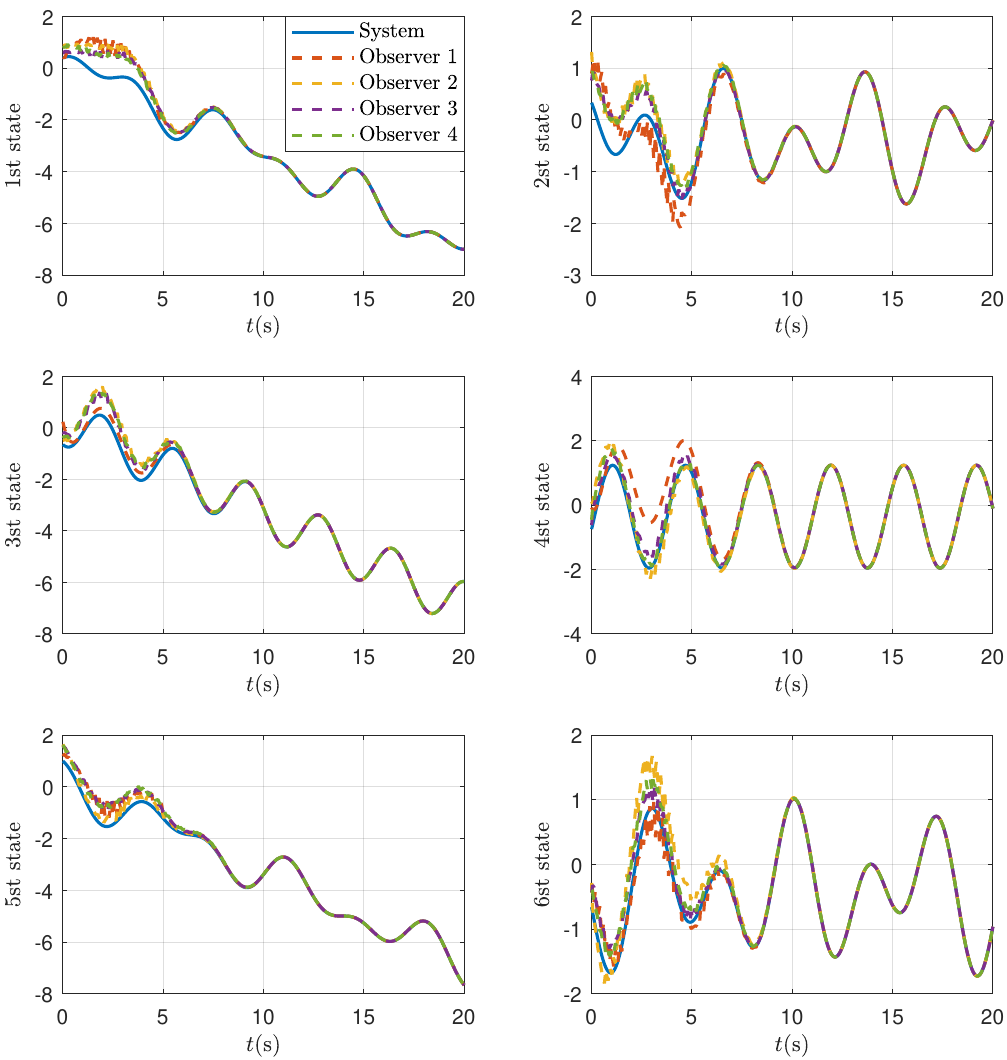}
	\caption{States of the system and the observers (node-based).}
	\label{State trajectory}
\end{figure}

\begin{figure}[htb]
	\centering
	\includegraphics[ width=\columnwidth]{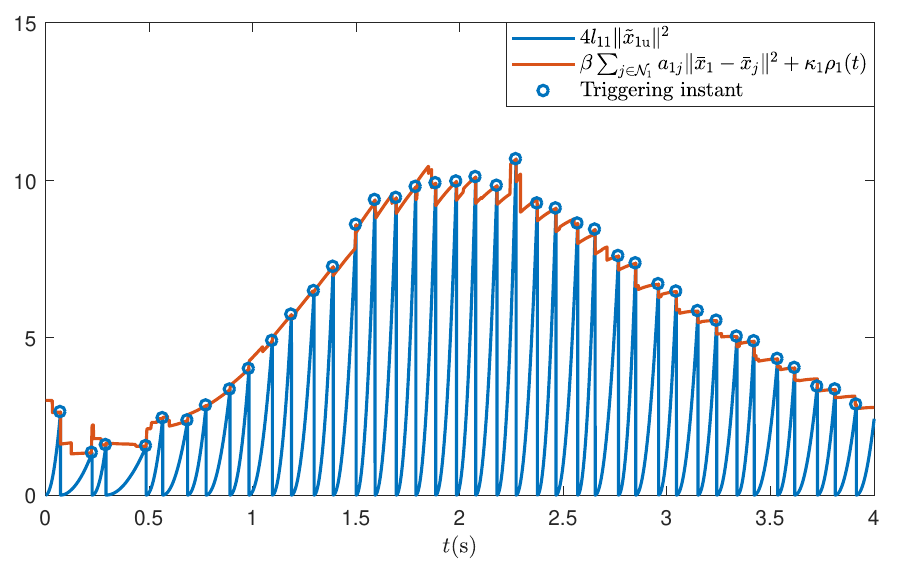}
	\caption{Triggering function of observer $1$ over $0-4$ seconds.}
	\label{triggering function n fig}
\end{figure}

\begin{figure}[htb]
	\centering
	\includegraphics[ width=\columnwidth]{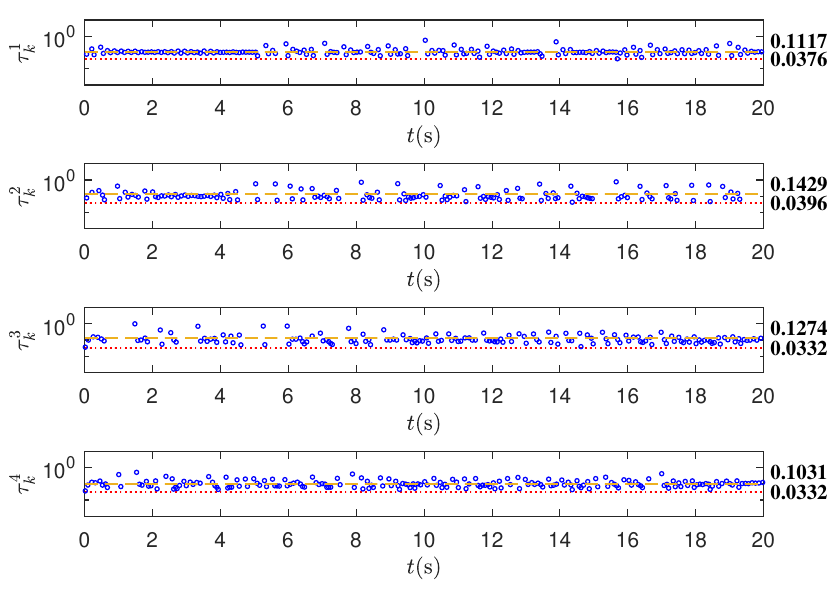}
	\caption{Event intervals $\tau_k^i$, $k\in\mathbb{N}$, for each agent. The dash and dotted lines together with the corresponding numbers indicate the mean and minimum of $\tau_k^i$ over $k$, respectively. The horizontal and vertical coordinates of each circle ``$\circ$'' are the event instants and event intervals, respectively.}
	\label{interval}
\end{figure}

\begin{figure}[htbp]
	\centering
	\includegraphics[ width=\columnwidth]{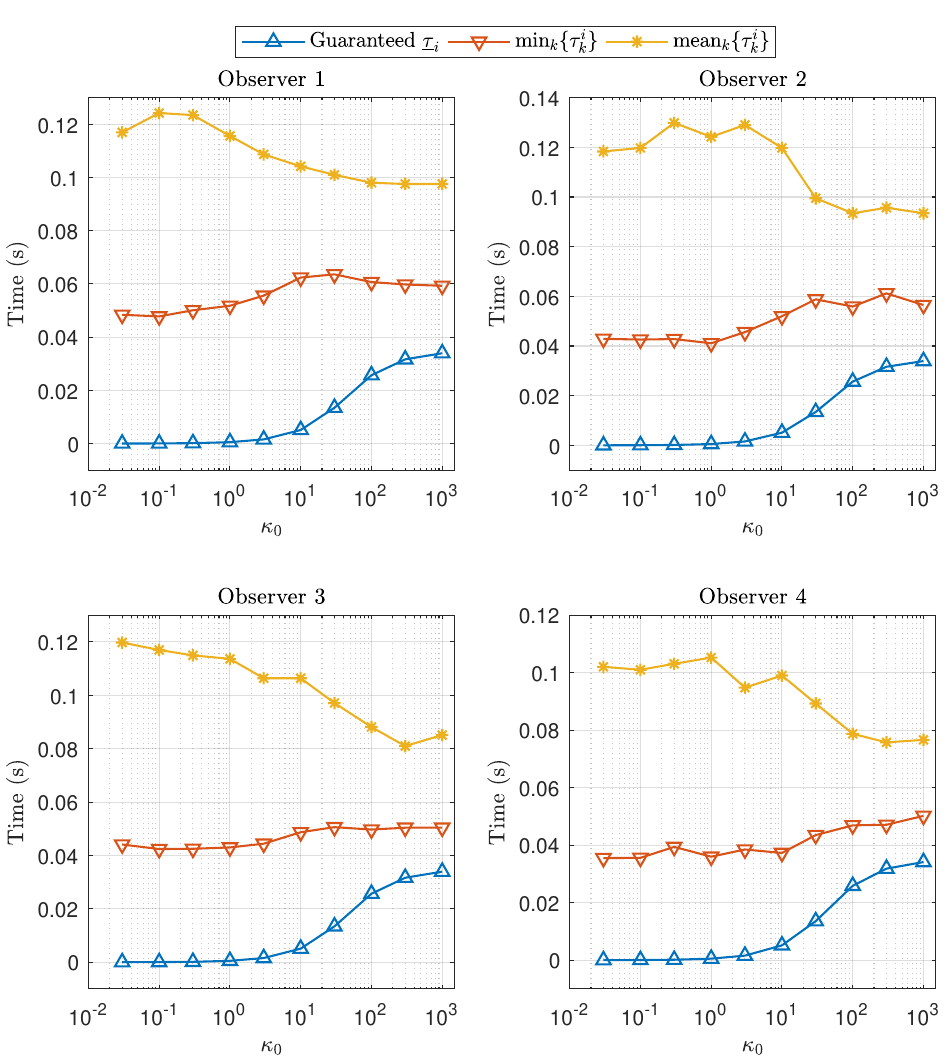}
	\caption{Guaranteed level $\underline{\tau}_i$, $\min_k\{\tau_k^i\}$ and $\mathrm{mean}_k\{\tau_k^i\}$ for different $\kappa_0$.}
	\label{lower bound of MIETs}
\end{figure}
	
	\textit{Example 1 (Node-based Event-Triggering Mechanism):} For the node-based event-triggered distributed observers in Section \ref{Node-based Event-Triggered Distributed Observer}, we set $c = 10.3$, $\beta = 0.9$, $\rho_i(0) = 1$, $\kappa_i = \kappa_0 = 0.03$, $\delta_i = 2$ and $\gamma_i = 10$ for $i\in\mathcal{V}$. Simulation results are shown in Figs. \ref{State trajectory}--\ref{lower bound of MIETs}. Figure \ref{State trajectory} shows that each state component of the system can be asymptotically estimated by the designed local observers. Figure \ref{triggering function n fig} illustrates the evolution of observer $1$'s triggering function (\ref{ET function n}), which demonstrates that the designed triggering mechanism (\ref{ET condition n}) is correctly executed. Figure \ref{interval} plots the inter-event times for each agent, showing that the event-triggering mechanism can guarantee strictly positive MIETs. Note that we set $\kappa_i = \kappa_0$ for $i\in\mathcal{V}$. Then, we assign different values to $\kappa_0$ to illustrate the effects of $\kappa_i$. Figure \ref{lower bound of MIETs} shows the guaranteed level $\underline{\tau}_i$ of MIETs, $\min_k\{\tau_k^i\}$ and $\mathrm{mean}_k\{\tau_k^i\}$ under different $\kappa_0$, where $\min_k\{\tau_k^i\}$ and $\mathrm{mean}_k\{\tau_k^i\}$ denote the minimum and the average of the sequence $\{\tau_k^i\}$, respectively. As can be seen from Fig. \ref{lower bound of MIETs}, the guaranteed level $\underline{\tau}_i$ increases with $\kappa_0$, which indicate that a larger $\kappa_0$ leads to larger lower bounds of MIETs. It can also be concluded from Fig. \ref{lower bound of MIETs} that the theoretical lower bound of MIET $\underline{\tau}_i$ calculated from (\ref{theoretical lower bound of MIET n}) is effective.

		\begin{figure}[htbp]
		\centering
		\includegraphics[ width=\columnwidth]{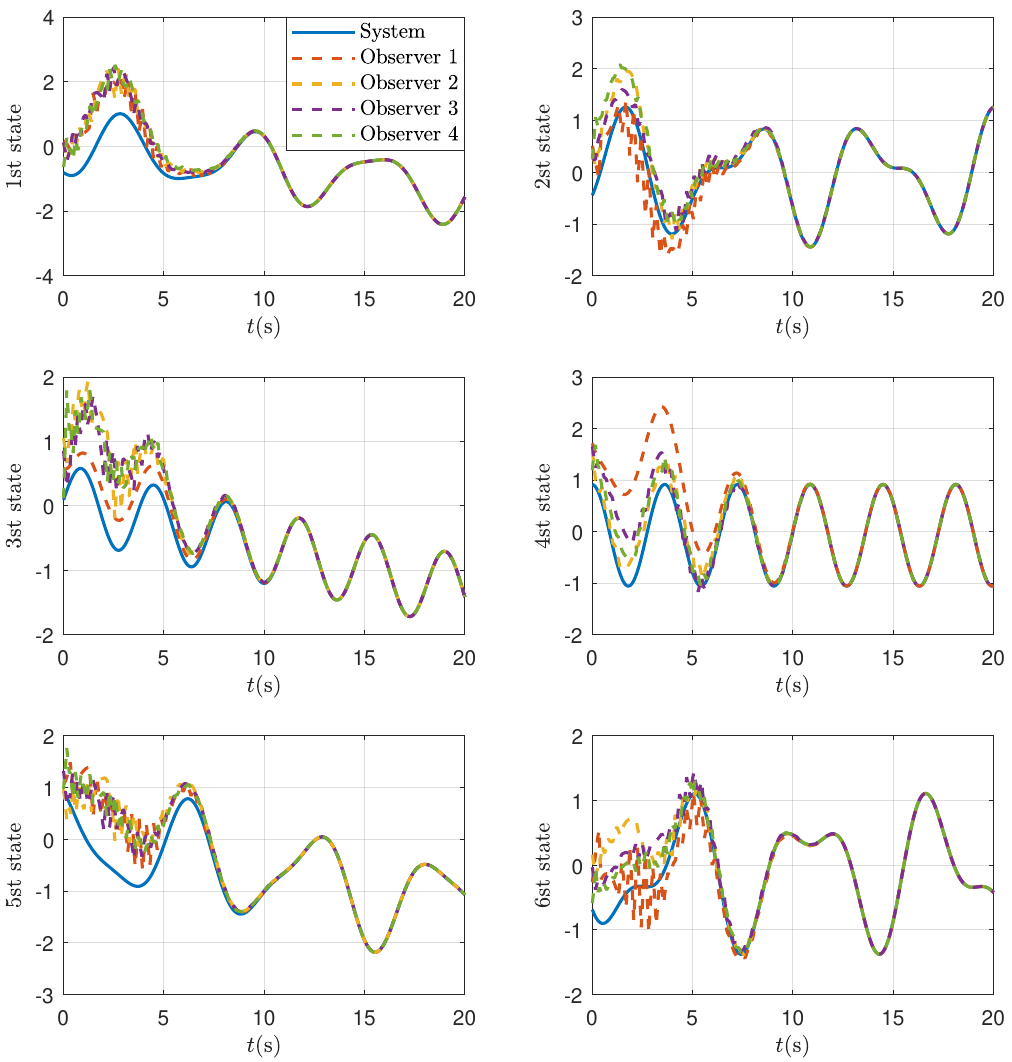}
		\caption{States of the system and the observers (edge-based).}
		\label{State trajectory e}
	\end{figure}

	\begin{table}[htbp]
		\centering
		\caption{The theoretical guaranteed level $\underline{\tau}_{ij}$ and the actual MIET $\min_k\{\tau_k^{ij}\}$}
		\label{conservative study}
		\renewcommand{\arraystretch}{1.2}
		\begin{tabular}{ccc} 
			\toprule
			Edge $(i,j)$ & $\underline{\tau}_{ij}$ ($\mathrm{s}$)& $\min_k\{\tau_k^{ij}\}$ ($\mathrm{s}$) \\
			\midrule
			(1,2) & 0.054 & 0.113 \\
			(2,1) & 0.053 & 0.099 \\
			(1,3) & 0.054 & 0.112 \\
			(3,1) & 0.053 & 0.107 \\		
			(2,4) & 0.053 & 0.099 \\		
			(4,2) & 0.060 & 0.093 \\
			(3,4) & 0.053 & 0.109 \\		
			(4,3) & 0.060 & 0.090 \\
			\bottomrule
		\end{tabular}
	\end{table}
	
	\textit{Example 2 (Edge-based Event-Triggering Mechanism):} For the edge-based event-triggered distributed observers in Section \ref{Edge-based Event-Triggered Distributed Observer}, we set $\check{c} = 5.3$, $\epsilon=0.01$, $\check{\beta} = 0.9$, $\rho_{ij}(0) = 1$, $\kappa_{ij} = 100$, $\delta_{ij} = 1$ and $\gamma_{ij} = 1$ for $(i,j)\in\mathcal{E}$. Simulation results are shown in Fig. \ref{State trajectory e} and Table \ref{conservative study}. From Fig. \ref{State trajectory e}, it can be seen that the system states can be asymptotically estimated by the edge-based event-triggered distributed observers. Table \ref{conservative study} shows the guaranteed level $\underline{\tau}_{ij}$ and the actual MIET $\min_k\{\tau_k^{ij}\}$ for each edge, which demonstrates that the obtained theoretical level $\underline{\tau}_{ij}$ provides an effective lower bound of the inter-event time.
	
	\begin{table}[htbp]
		\centering
		\caption{Comparison of theoretical lower bounds of MIETs}
		\label{comparative study}
		\renewcommand{\arraystretch}{1.2}
		\begin{tabular}{ccc} 
			\toprule
			Agent $i$ & Ours & \cite{zhu2025hybrid}  \\
			\midrule
			1 & $0.0131$ $\mathrm{s}$& $4.34 \times 10^{-7}$ $\mathrm{s}$\\
			2 & $0.0132$ $\mathrm{s}$& $4.34 \times 10^{-7}$ $\mathrm{s}$\\
			3 & $0.0131$ $\mathrm{s}$& $4.34 \times 10^{-7}$ $\mathrm{s}$\\
			4 & $0.0133$ $\mathrm{s}$& $4.34 \times 10^{-7}$ $\mathrm{s}$\\		
			\bottomrule
		\end{tabular}
	\end{table}	
	
	\textit{Example 3 (Comparison with \cite{zhu2025hybrid} and \cite{li2025distributed}):} For the node-based event-triggered mechanism, we compare our method with \cite{zhu2025hybrid}, focusing primarily on the theoretical lower bounds of MIETs. For parity, the relevant parameters in \cite{zhu2025hybrid} are set as $L_\omega^i=4$ and  $\lambda_i = 0.25$ for $i\in\mathcal{V}$. We set the parameters in (\ref{choice of c n}), (\ref{ET function n}) and (\ref{dynamic variable n}) as $c = 10.3$, $\beta = 0.9$, $\rho_i(0) = 1$, $\kappa_i = \kappa_0 = 30$, $\delta_i = 6$ and $\gamma_i = 10$ for $i\in\mathcal{V}$. The results are shown in Table \ref{comparative study}, from which it can be seen that the theoretical lower bounds of MIETs in \cite{zhu2025hybrid} are more conservative than ours. For the edge-based event-triggering mechanism, with the system matrix and measurement matrix set as in (\ref{A sim}) and (\ref{H sim}), the method in \cite{li2025distributed} becomes inapplicable because the LMI (9) in \cite{li2025distributed} is infeasible.

	\section{Conclusion} \label{Conclusion}
	This paper investigates the distributed state estimation problem under event-triggered communication. We propose a dynamic event-triggered distributed observer that enables each agent to reconstruct the system state in a distributed manner. Our work focuses on the design of  dynamic event-triggering mechanisms and the analysis of the resulting inter-event intervals. We prove that the dynamic event-triggering mechanism can guarantee positive MIETs for distributed observers. The exponential convergence of the estimation error is also proven. Besides, both node-based and edge-based event-triggering mechanisms are proposed. Note that this paper only considers an undirected communication graph. Therefore, the extension to directed communication topologies is an interesting direction for future research.

\bibliographystyle{plain}        
\bibliography{Ref}           

\end{document}